\documentclass[12pt, draftclsnofoot, onecolumn]{IEEEtran}
 \usepackage{amsmath,amssymb}
 \usepackage{subfigure}
 \usepackage{graphicx,graphics,color,psfrag}
 \usepackage{cite,balance}
 \usepackage{caption}
 \allowdisplaybreaks
 \usepackage{algorithm}
  \usepackage{algorithmic}
 \usepackage{accents}
 \usepackage{amsthm}
 \usepackage{bm}
  \usepackage{url}
 \usepackage[english]{babel}
 \usepackage{multirow}
 \usepackage{enumerate}
 \usepackage{cases}
 \usepackage{stfloats}
 \usepackage{dsfont}
 \usepackage{color,soul}
 \usepackage{amsfonts}
 \usepackage{cite,graphicx,amsmath,amssymb}
 \usepackage{subfigure}
 \usepackage{fancyhdr}
 \usepackage{hhline}
 \usepackage{graphicx,graphics}
 \usepackage{array,color}
 \usepackage{amsmath}

\newtheorem{lemma}{\emph{\underline{Lemma}}}

\newtheorem{remark}{\bf \emph{\underline{Remark}}}
\def\phi{\varphi}

\def\l{\left}
\def\r{\right}
\def\({\left(}
\def\){\right)}

\def\b0{{\mathbf{0}}}

\newcommand{\nn}{\nonumber}

\begin{document}
\captionsetup[figure]{name={Fig.}}
\title{\huge  
Hybrid Offline-Online Design for UAV-Enabled Data Harvesting in Probabilistic LoS Channel}
\author{Changsheng You, \IEEEmembership{Member,~IEEE} and Rui Zhang, \emph{Fellow, IEEE}   \thanks{\noindent Part of this work has been presented in IEEE Global Communications Conference, Waikoloa, HI, USA, 2019 \cite{you2019globecom}. \newline \indent C. You and R. Zhang are  with the Dept. of Electrical and Computer Engineering, National University of Singapore, Singapore (Email: eleyouc@nus.edu.sg, elezhang@nus.edu.sg).
 }}

\maketitle
\vspace{-15pt}
\begin{abstract}

This paper considers an unmanned aerial vehicle (UAV)-enabled wireless sensor network (WSN) in urban areas, where a {\color{black}rotary-wing} UAV is deployed  to collect data from distributed sensor nodes (SNs) within a given duration. To characterize the occasional building blockage between the UAV and SNs, we construct the  \emph{probabilistic} line-of-sight (LoS) channel model for a Manhattan-type city by using the combined simulation and data regression method, which is shown in the form of a \emph{generalized logistic} function of the UAV-SN elevation angle. We assume that only the knowledge of SNs' locations and the probabilistic LoS channel model is known \emph{a priori}, while the UAV can obtain the \emph{instantaneous} LoS/Non-LoS channel state information (CSI) with the SNs in \emph{real time} along its flight. Our objective is to maximize the minimum (average) data collection rate from all the SNs for the UAV. To this end, 
we formulate a new  rate maximization problem by jointly optimizing the UAV \emph{three-dimensional} (3D) trajectory and transmission scheduling of SNs. Although the optimal solution is intractable due to the lack of the complete UAV-SNs CSI,  we  propose in this paper a \emph{novel} and \emph{general} design method, called \emph{hybrid offline-online} optimization,  to obtain a suboptimal solution to it, by leveraging both the statistical and real-time CSI. Essentially, our proposed method  \emph{decouples} the  joint design of UAV trajectory and communication scheduling into two phases: namely, an \emph{offline} phase that determines the UAV path prior to its flight based on the probabilistic LoS channel model, followed by an \emph{online} phase that adaptively adjusts the UAV flying speeds along the offline optimized path as well as communication scheduling based on the instantaneous UAV-SNs CSI and SNs' individual amounts of data  received accumulatively. Extensive simulation results are provided to show the significant rate performance improvement of our proposed design as compared to various benchmark schemes.
\end{abstract}
\begin{IEEEkeywords}
UAV communications, wireless sensor network, 3D trajectory optimization, probabilistic LoS channel, hybrid offline-online design.
\end{IEEEkeywords}
\section{Introduction}

The vision  of Internet-of-Drones (IoD) has spurred intensive enthusiasm in recent years on deploying unmanned aerial vehicles (UAVs)  (or Drones) to automate a proliferation of applications, such as
aerial inspection, photography, packet delivery, remote sensing, and so on  \cite{zeng2019accessing,zeng2016wireless}.  Particularly for  wireless communications, the unique features of UAVs such as high mobility, controllably maneuver as well as LoS-dominant air-ground channels have incentivized   both academia and industry to integrate them into  the conventional terrestrial  wireless networks for enhancing their  coverage and throughput, leading to various new applications, such as UAV-assisted terrestrial communications \cite{zeng2016wireless,wu2018joint,lyu2017placement,mozaffari2016efficient,mozaffari2016unmanned,bor2016efficient}, cellular-connected UAVs \cite{zhang2018cellular,zeng2019cellular,lyu2019network}, UAV-enabled mobile  relaying \cite{zeng2016throughput,chen2018local},  UAV-enabled wireless sensor networks (WSNs) \cite{you20193d,zhan2018energy,ebrahimi2018uav,gong2018flight,liu2018age,abd2018average,abd2019deep}, to name a few.
{\color{black}Specifically, for the UAV-enabled WSNs that utilize UAVs as mobile data collectors to directly receive data from spatially-separated SNs, one key problem is to design the UAV trajectory in the \emph{three-dimensional} (3D) space for maximizing data harvesting throughput for delay-tolerant applications \cite{you20193d} or minimizing data collection time or age-of-information (that characterizes the freshness of information) for delay-sensitive applications \cite{gong2018flight,liu2018age,abd2018average,abd2019deep}.} Although there are prior works that addressed  this problem (see e.g., \cite{zhan2018energy,gong2018flight,ebrahimi2018uav,liu2018age,abd2018average,abd2019deep}), they mostly  adopted the \emph{deterministic line-of-sight (LoS)-dominant} channel model which is usually a valid assumption for rural areas without high and dense obstacles. As a result, these works usually considered the design of  two-dimensional (2D) UAV trajectories only  with fixed (minimum) UAV flying altitude in an \emph{offline} manner. 

Such a design approach, however, has two main limitations. First,  in  urban areas with typically high and dense buildings/obstacles, the simplified LoS-dominant  channel model  can be practically inaccurate, as they do not capture the critical effects of  UAV \emph{location-dependent} multi-path fading and shadowing. To address this issue, two more sophisticated UAV-ground  channel models have been proposed in the literature  to improve accuracy, namely, the \emph{(elevation) angle-dependent Rician fading} and the \emph{probabilistic LoS} channel models. Specifically, when the UAV flies sufficiently high above the ground, the shadowing effect diminishes and the main channel randomness comes from multi-path reflection, scattering, and diffraction by the ground obstacles. Such characteristics can be captured by the angle-dependent Rician fading channel  model \cite{azari2018ultra}, where the Rician factor generally increases with the elevation angle between the UAV and its served ground node. Based on this model, 3D UAV trajectory has been designed  in \cite{you20193d} for maximizing the data harvesting throughput in UAV-enabled WSNs. In contrast, if the UAV flies at a relatively low altitude, the shadowing effect becomes more significant, due to which the signal propagation between the UAV and a ground node can be occasionally blocked by buildings, where the likelihood of blockage in general depends on the relative position between the UAV and ground node, as well as  the distributions of building density and  height. Roughly speaking, 
the UAV-ground  channel can be divided into two states, namely, LoS versus non-LoS (NLoS), each characterized by a different model. To avoid the excessive measurements for obtaining the complete information of LoS/NLoS channels at each location in a large geographical area, the probabilistic LoS channel model has been proposed  in \cite{al2014optimal} to characterize the channel state statistically by modeling the occurrence probabilities of the LoS/NLoS states via  heuristic  functions of the UAV-ground elevation angle. Intuitively, the LoS probability in this model increases with the UAV-ground elevation angle, by either moving the UAV horizontally closer to the ground node or increasing its altitude above the ground; while in the latter case, the channel path loss also increases with distance, thus yielding an interesting  \emph{angle-distance trade-off} in the UAV-ground channel gain versus its altitude, as will be further investigated  in this paper.

The second limitation is that, the adopted offline design approach for the UAV trajectory and communication scheduling based on the deterministic LoS-dominant channel model may suffer considerable (rate) performance loss in the urban areas with random building blockage, since the offline designed policy \emph{cannot} adapt to the \emph{real-time} location-dependent  UAV-ground channel states, which is rather critical due to the significant disparity between the channel strengths under the LoS and NLoS states. Although some recent works have adopted the probabilistic LoS channel model for designing UAV trajectory (e.g., \cite{esrafilian2018learning,zeng2019energy}),  they still followed the offline design approach by considering the \emph{deterministic expected} channel gain for the random (uncertain) channel state and hence did not involve channel-aware online  adaptation. To tackle this issue, an initial attempt has been made in \cite{chen2018local} where the authors proposed a nested segmented UAV-ground channel model and developed a customized online algorithm to search the optimal UAV position for UAV relaying by leveraging the information of local channel state and terrain topology. Recently, a new reinforcement learning-based UAV path design was developed in \cite{zeng2019path} that progressively determined UAV trajectory according to real-time channel measurements. This approach, however, is \emph{data expensive} in the sense that it entails abundant real UAV flight data   in the offline learning phase. In addition, it is worth mentioning that there has been a number of recent works that used deep learning to online design the  UAV trajectory, but they primarily targeted  to learn optimization solutions  \cite{liu2018energy} or adapt to other environmental randomness (instead of the uncertain channel state) such as intermittent interference \cite{challita2019interference} and random user movement \cite{liu2018trajectory}. Besides channel-state awareness, another key issue for the online UAV trajectory design is affordable computational complexity  for practical implementation. This issue has been widely investigated in the conventional UAV trajectory design for obstacle avoidance.  For example,  a receding-horizon-control based deterministic path planning was studied in \cite{kuwata2003real} that progressively plans the overall trajectory by finding the local trajectory in a finite  forward time horizon. In addition, randomized path planning has been applied  by using heuristic functions (called potential fields) to guide the path search or leveraging a roadmap that contains pre-computed feasible paths for the path selection  \cite{barraquand1997random}. Nevertheless, these approaches are generally heuristic and cannot be directly applied to the new  \emph{communication-aware} UAV trajectory design due to the more complicated coupling  between the communication performance  and 3D UAV trajectory.

Motived by the above, this  paper aims to overcome the aforementioned limitations in the existing designs for  communication-aware UAV trajectory optimization. For the purpose of exposition, we consider a UAV-enabled WSN where one single {\color{black}rotary-wing} UAV flies over multiple sensor nodes (SNs)  to collect data from them within a given duration. The SNs are normally in the silent mode for energy saving and transmit data only when being waken up by the UAV (e.g., via a beacon signal broadcast by the UAV). Assume that the UAV only has the knowledge of SNs' locations and the probabilistic LoS channel model prior to its flight, while  it can obtain the instantaneous UAV-SNs channel state information (CSI) along its flight. Our objective is to maximize the minimum (average) data collection rate from all the SNs for the UAV by jointly designing its 3D trajectory and transmission scheduling of SNs. The main contributions of this paper are summarized as follows.

\begin{itemize}
\item Firstly, we propose a \emph{novel} and \emph{general} method to design the 3D UAV trajectory and communication scheduling adaptive to the random building blockage in urban areas. To this end, we start with improving the accuracy of the conventional probabilistic LoS  channel model by applying the combined simulation and data regression method. The newly constructed model for a Manhattan-type city is shown to be a \emph{generalized logistic} function of the UAV-SN elevation angle. Based on this model,  we then formulate an optimization problem to maximize the minimum (average) data collection rate from all the SNs for the UAV, whose optimal solution, however, is difficult to obtain due to the lack of the complete UAV-SNs CSI at all possible UAV locations. {\color{black}To tackle this difficulty, we propose to derive its suboptimal solution based on a new \emph{hybrid offline-online optimization} method, by leveraging both the statistical and real-time CSI, which greatly differs from the conventionally adopted offline design approach in e.g., \cite{you20193d,esrafilian2018learning}.}  The main idea of our proposed method is to decouple the joint design of UAV trajectory and communication scheduling into two phases: namely, an offline phase that optimizes the UAV path (i.e., specifying the UAV flying direction via a sequence of ordered waypoints along the trajectory) prior to its flight based on the probabilistic LoS channel model, followed by an online phase that adaptively adjusts the UAV flying speeds along the offline optimized path as well as its communication scheduling with SNs based on the instantaneous UAV-SNs CSI and SNs' individual amounts of data  received accumulatively.
\item Secondly, we propose efficient algorithms for solving the formulated optimization problems in  both the offline and online phases. Specifically, in the offline phase, we aim to maximize the minimum \emph{expected} (average) rate from all the SNs based on the probabilistic LoS channel model. To solve this non-convex problem, we approximate the expected rate function, which is highly complicated with respect to (w.r.t.) the 3D UAV trajectory, by a tractable lower bound based on the dominant rate in the LoS channel state. Since the reformulated problem is still non-convex and thus is difficult to solve, we further apply continuous relaxation to the integer communication scheduling constraints in the problem and then solve the relaxed problem sub-optimally by using the techniques of block coordinate descent (BCD) and successive convex approximation (SCA).\footnote{\color{black}{The detailed offline design in this paper is different from that in \cite{you20193d}, since we consider a new UAV-ground channel model and derive a new expression for the expected-rate function.}} On the other hand, for the online phase, we formulate a linear programming (LP) to maximize the \emph{updated} minimum expected rate from all the SNs at each waypoint along the offline optimized path. The LP can be efficiently solved with low complexity at the UAV in real time, thus making the online adaptation amenable  to practical implementation. Extensive simulation results are provided to verify the effectiveness of the proposed hybrid design.
\item {\color{black}Thirdly, we draw several interesting observations and  insights into the proposed hybrid design. As compared to the conventional 2D trajectory design based on the simplified LoS channel model, our proposed 3D UAV trajectory can further exploit the additional degrees-of-freedom (DoF) of  the UAV vertical trajectory to balance the said angle-distance trade-off for rate enhancement.
Moreover, the proposed low-complexity online adaptation can effectively leverage the real-time CSI to improve the minimum-rate performance, by dynamically  scheduling the SNs with favorable channels for data transmission as well as adjusting the UAV flying speeds.}
\end{itemize}

The remainder of this paper is organized as follows. Section~\ref{Sec:Model} introduces the system model, based on which we formulate an optimization problem in Section~\ref{Sec:Problem} and present the main idea of the proposed hybrid offline-online optimization method for solving it. Then the offline and online phases are designed in Sections~\ref{Sec:Off} and \ref{Sec:On}, respectively. Last, extensive simulation results and discussions are provided in Section~\ref{Sec:Simu}, followed by the conclusions given in Section~\ref{Sec:Conc}.

\vspace{-5pt}
\section{System Model}\label{Sec:Model}

{\color{black}Consider a UAV-enabled WSN where a rotary-wing UAV is dispatched to collect data from $K$ ground SNs, denoted by the set $\mathcal{K}=\{1,\cdots, K\}$, within a given  duration of $T_0$.}
 The SNs' locations are represented by $({\mathbf{w}}_k^{T}, 0)$, $\forall k\in \mathcal{K}$, where ${\mathbf{w}}_k=[x_k, y_k]^{T}\in\mathbb{R}^{2\times1}$ denotes the horizontal coordinate of SN $k$. In the following subsections, the models of UAV trajectory, UAV-SN channel, and data collection from SNs  are described, respectively.

\vspace{-5pt}
\subsection{UAV Trajectory Model}\label{d}
{\color{black}For ease of analysis, the time horizon $T_0$ is partitioned (discretized)  into $N$ equal time slots with sufficiently small slot length $\delta=T_0/N$ such that the UAV's location can be assumed to be approximately unchanged relative to the ground SNs  within each time slot.  Moreover, we assume that the UAV's initial and final locations are predetermined  when e.g., the UAV is set to be launched and landed at certain preferred locations or the data-collection mission specifies the initial and final locations.\footnote{{\color{black}For the UAV with unconstrained initial and/or final locations as well as finite flight duration, the optimal UAV's initial/finial location for maximizing the achievable rate of the system in general can only be obtained by an exhaustive search over the 3D space of interest.}}}
As such, the UAV trajectory can be approximated by an $(N+1)$-length 3D  sequence $\{({\mathbf{q}}_n^T, z_n)\}_{n=1}^{N+1}$ with $({\mathbf q}_1^T, z_1)=({\mathbf{q}}_I^{T}, z_I)$ and $({\mathbf q}_{N+1}^T, z_{N+1})=({\mathbf{q}}_F^{T}, z_F)$ denoting the UAV's initial and final locations, respectively. {\color{black}Assuming  that the UAV can independently control its horizontal and vertical flying speeds subject to their maximum values, denoted  by $V_{\rm{xy},\max}$ and $V_{\rm{z},\max}$  respectively in meter/second (m/s)\footnote{\color{black}\url{https://www.dji.com/sg/phantom-4-rtk/info}}}, then  the maximum horizontal and vertical flying distances within each time slot are given by $S_{\rm{xy},\max}=V_{\rm{xy},\max}\delta$ and $S_{\rm{z},\max}=V_{\rm{z},\max}\delta$, leading to the following trajectory constraints\footnote{{\color{black}For the UAV with a 3D maximum compound flying speed $V_{\max}$, the assumption of  independently-controlled UAV horizontal and vertical flying speeds  still applies as long as $V_{\rm{xy},\max}^2+V_{\rm{z},\max}\le V_{\max}^2$.}}
\vspace{-5pt}
\begin{align}
||{\mathbf q}_{n+1}-{\mathbf q}_n||\le S_{\rm{xy},\max},~~ |z_{n+1}-z_n|\le S_{\rm{z},\max}, ~~\forall n\in\mathcal{N},
\end{align}
where $\mathcal{N}=\{1, \cdots, N\}$.  To avoid obstacles such as buildings and conform to aerial regulations, the UAV is required to fly at an altitude within a given range, yielding the following constraints
\begin{equation}\label{Eq:zmin}
H_{\min}\le z_n\le H_{\max}, \quad\forall n \in\mathcal{N}.
\end{equation}
{\color{black}In practice, to guarantee a certain accuracy on the time-discretization approximation, the number of time slots $N$ can be chosen to satisfy that $\frac{\max\{S_{\rm{xy},\max},S_{\rm{z},\max}\}}{H_{\min}}\le\varepsilon_{\max}$, where $\varepsilon_{\max}$ is a given threshold. Thus, the number of time slots $N$ can be set as $N=\left\lceil\frac{T_0\max\{V_{{\rm xy},\max},V_{{\rm z},\max}\}}{H_{\min}\varepsilon_{\max}}\right\rceil$, where $\lceil\cdot\rceil$ denotes the ceiling operation, since further increasing $N$ will increase the design complexity for the 3D UAV trajectory.}

\begin{figure}[t!]
\centering
\subfigure[A Manhattan-type city.]{\label{FigCity}
\includegraphics[height=5.3cm]{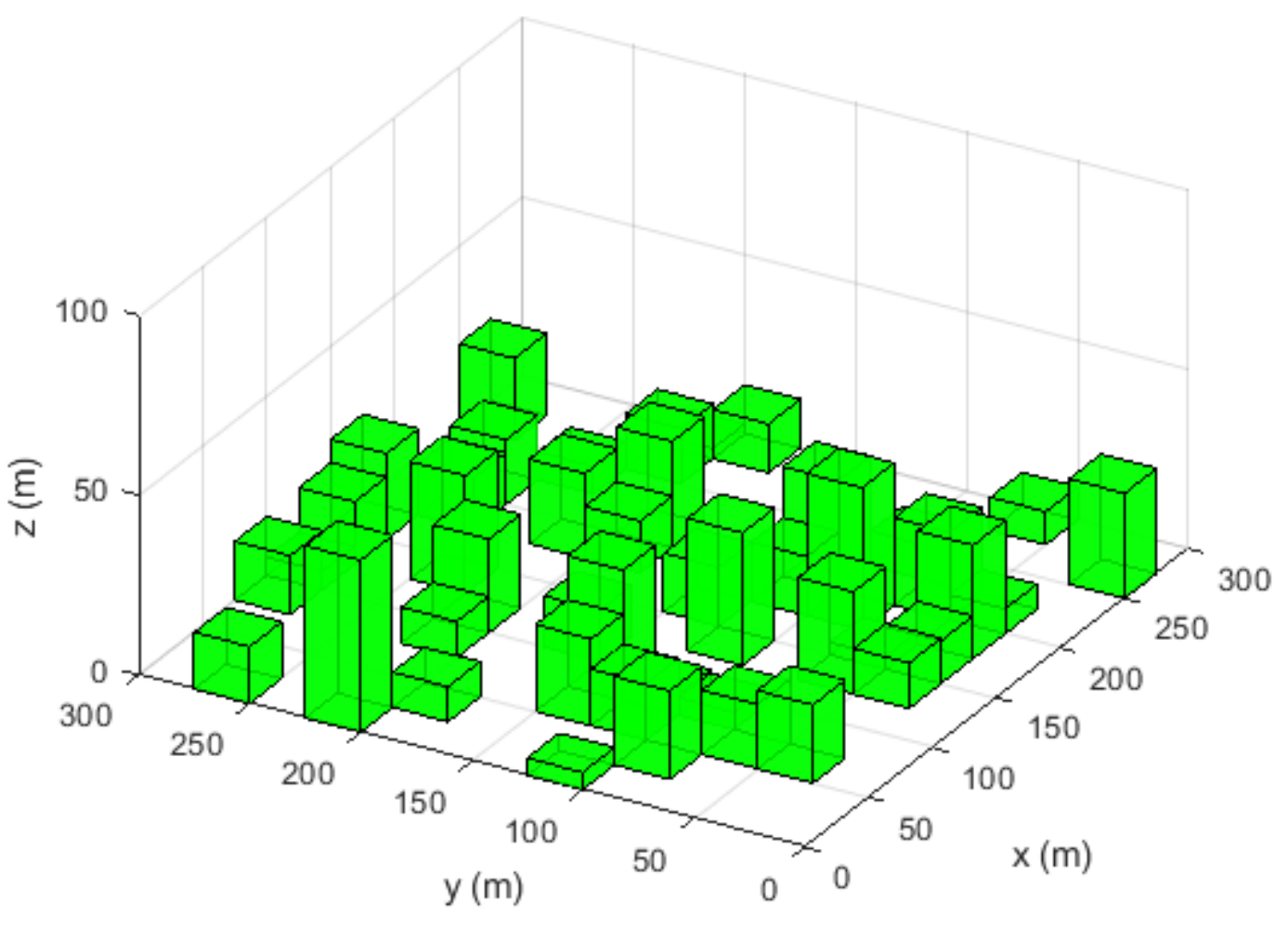}}
\hspace{6mm}
\subfigure[Generalized logistic regression for approximating \newline\indent \quad~ LoS probability.]{\label{Fig:LoSProb}
\includegraphics[height=5.5cm]{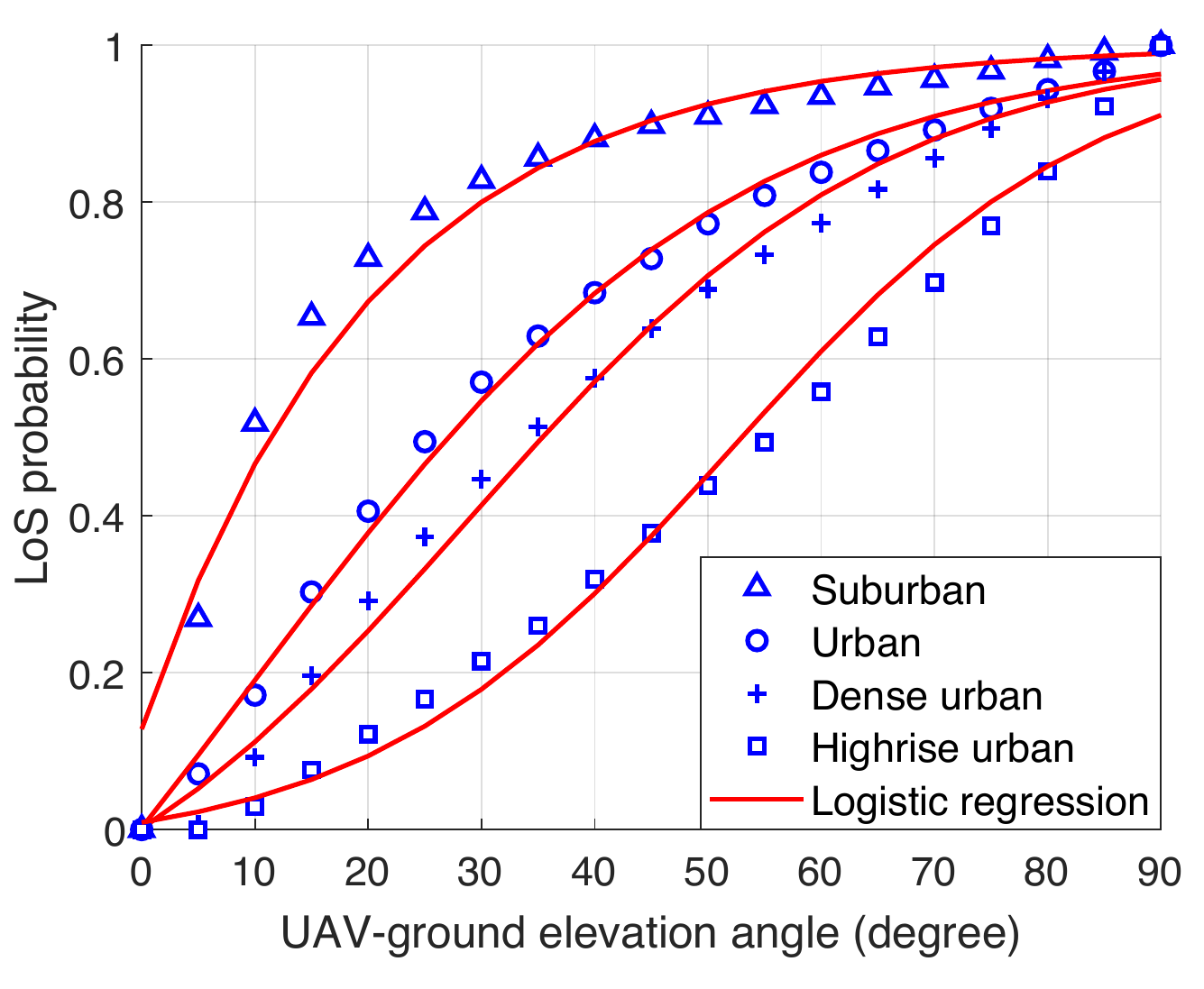}}
\centering
\caption{Generalized logistic model for the LoS probability in a Manhattan-type city.}
\label{Fig:TrajTime}
\end{figure} 


\subsection{UAV-SN Channel Model}\label{Sec:ChanMod}
Assuming that the complete CSI with SNs for the UAV  in a given  3D space, is not known \emph{a priori}, we characterize the statistical UAV-SN channel model as follows by accounting for occasional building blockage. For each SN $k$, let $c_{k,n}$ denote the binary UAV-SN channel state in each time slot $n$, where $c_{k,n}=1$  and $c_{k,n}=0$ represent respectively the LoS and NLoS states. Following the probabilistic LoS channel model \cite{al2014optimal}, the uncertain channel state $c_{k,n}$ is assumed to be \emph{independently} distributed over different time slots, {\color{black}{which is practically valid when the number of time slots is sufficiently large since the actual UAV-ground CSI correlation over time  can be averaged out.}}\footnote{\color{black}{The actual UAV-ground CSI correlation will be taken into account in the proposed online design.}} Moreover,  in each time slot, the  LoS probability, denoted by $\mathbb{P}(c_{k,n}=1)$, is a function of the UAV-SN elevation angle.  Among others, one commonly used model to  approximate the LoS probability is based on a \emph{simple} logistic function in the form of 
\begin{equation}
\mathbb{P}(c_{k,n}=1)=\frac{1}{1+a e^{-b(\theta_{k,n}-a)}},\label{Eq:LoSProb}
\end{equation} where $\theta_{k,n}$ is the elevation angle between the UAV and SN $k$ in  time slot $n$, given by
 \begin{equation}
\theta_{k,n}=\frac{180}{\pi}\arctan\l(\frac{z_n}{|| {\mathbf{q}}_n-{\mathbf{w}}_k||}\r),\label{Eq:angle}
\end{equation} and  $a$ and $b$ are modeling parameters to be specified. However, this model is obtained by curve-fitting an approximate mathematical expression for the LoS probability under certain assumptions (e.g., buildings are evenly spaced between the UAV and SNs) \cite{data2003prediction} and thus may be practically \emph{inaccurate}. In this paper, we improve the accuracy of the probabilistic LoS channel model by applying the combined simulation and data regression method. To be specific, we firstly simulate a Manhattan-type city as shown in Fig.~\ref{FigCity} according to the typical parameters in built-up environments \cite{data2003prediction} and then compute the LoS probability under different UAV-SN elevation angles (see Appendix~\ref{App:SimuLoS} for  the detailed simulation method). Next, we apply the data regression method to seek a model that can well fit the simulation data for the LoS probability versus the elevation angle in different environments. As shown in Fig.~\ref{Fig:LoSProb}, the LoS probability, re-denoted by $P^{\rm{L}}_{k,n}\overset{\triangle}{=}\mathbb{P}(c_{k,n}=1)$ for convenience, can be largely approximated by a \emph{generalized logistic}  function in the following new  form
\vspace{-3pt}
\begin{equation}
P^{\rm{L}}_{k,n}=B_3+\frac{B_4}{1+e^{-(B_1+B_2\theta_{k,n})}},\label{Eq:LoSProb}
\end{equation} 
where  $B_1<0$, $B_2>0$,  $B_4>0$, and $B_3$ are constants with $B_3+B_4=1$, which are   determined by the specific environment.  The corresponding NLoS probability can be obtained as $P^{\rm{N}}_{k,n}\overset{\triangle}{=}\mathbb{P}(c_{k,n}=0)=1-P^{\rm{L}}_{k,n}$.  Then the large-scale channel power gain between the UAV and SN $k$ in each time slot $n$, including both the path loss and shadowing, can be approximately modeled by 
\begin{equation}\label{Eq:Channel}
h_{k,n}=c_{k,n} h_{k,n}^{\rm{L}}+(1-c_{k,n})h^{\rm{N}}_{k,n},
\end{equation}
 where 
 \begin{equation}\label{Eq:BinaryChannel}
 h_{k,n}^{\rm{L}}= \beta_0 d_{k,n}^{-\alpha_{\rm{L}}}, ~~ h^{\rm{N}}_{k,n}=\mu \beta_0  d_{k,n}^{-\alpha_{\rm{N}}}
 \end{equation}
 denote respectively the channel power gains conditioned on the LoS and NLoS states,
 $\beta_0$ is the average channel power gain at a reference distance of $d_0=1$ m in the LoS state, $\mu<1$ represents the additional signal attenuation factor due to the NLoS propagation, $\alpha_{\rm{L}}$ and $\alpha_{\rm{N}}$  denote respectively the average path loss exponents for the LoS and NLoS states, with $2\le \alpha_{\rm{L}} <\alpha_{\rm{N}}\le 6$ in practice, and \begin{equation}
d_{k,n}=\sqrt{|| {\mathbf{q}}_n-{\mathbf{w}}_k||^2+z_n^2} \label{Eq:Distance}
\end{equation} is the distance between the UAV and SN $k$ in time slot $n$.

\subsection{Data Collection Model}
{\color{black}We assume that the SNs are powered by energy harvesting and thus have sufficient energy for performing the best-effort data-transmission policy, so that they can send data to the UAV at their maximum transmit power $P_k$ in their scheduled transmission slots  and otherwise remain in the silent mode.}
 Let $a_{k,n}$ denote the binary communication scheduling variable for SN $k$ in time slot $n$, where SN $k$ transmits if  $a_{k,n}=1$ and keeps silent  otherwise. 
In each time slot, we assume that only one SN is scheduled for transmission, leading to the following  scheduling constraints
\begin{align}
&\sum_{k=1}^{K} a_{k,n}\le 1, \quad \forall n\in\mathcal{N},\qquad
 a_{k,n}\in\{0,1\},  \quad \forall k\in\mathcal{K},n\in\mathcal{N}.\label{Eq:P1InterCons}
\end{align}
If SN $k$ is scheduled, the corresponding maximum achievable rate from the SN, denoted by $r_{k,n}$ in bits/second/Hertz (bps/Hz), is given by
\begin{align}
r_{k,n}=\log_2\l(1+\dfrac{h_{k,n} P_k}{\sigma^2 \Gamma}\r),\label{Eq:Rate}
\end{align} 
where $h_{k,n}$ is the real-time (large-scale)  channel power gain given in \eqref{Eq:Channel}, $\sigma^2$ denotes the receiver noise power, and $\Gamma>1$ is the signal-to-noise ratio (SNR) gap between the practical modulation-and-coding  scheme and the theoretical Gaussian signaling.\footnote{For simplicity, we assume that each time slot consists of a large number of  fading blocks due to  small-scale fading, and their effects have been averaged out in each time slot by employing a sufficiently long channel code; thus, the rate approximation given in \eqref{Eq:Rate} is practically valid. } Combining \eqref{Eq:Rate} and \eqref{Eq:Channel}--\eqref{Eq:BinaryChannel}   yields the following real-time \emph{channel-state-dependent} achievable rate
\begin{equation}
r_{k,n}=c_{k,n} r^{\rm{L}}_{k,n}+(1-c_{k,n})r^{\rm{N}}_{k,n},\label{Eq:StateRate}
\end{equation}
where 
\begin{align}
r^{\rm{L}}_{k,n}=\log_2\l(1+\dfrac{\gamma_k}{d_{k,n}^{\alpha_{\rm{L}}}}\r), ~~
r^{\rm{N}}_{k,n}=\log_2\l(1+\dfrac{\mu \gamma_k}{d_{k,n}^{\alpha_{\rm{N}}}}\r)\label{Eq:LNrate}
\end{align}
denote respectively the achievable rates conditioned on the LoS and NLoS states, and $\gamma_k=\frac{\beta_0 P_k}{\sigma^2 \Gamma}$.

\section{Problem Formulation and Proposed Hybrid Design}\label{Sec:Problem}

Consider the UAV-enabled data collection in a Manhattan-type city where the buildings are randomly and uniformly generated as described in Appendix~\ref{App:SimuLoS}. We assume that, prior to the UAV's flight, only the knowledge of SNs' locations and the probabilistic LoS channel model is known, while the UAV can estimate the instantaneous CSI perfectly  with individual SNs in real time along its flight.\footnote{In practice, the UAV can only estimate the CSI by receiving signals from the SNs within its communication coverage. Nevertheless, our assumption of the UAV knowing the CSI  with all SNs does not compromise the above practicability  since the SNs far away from the UAV are expected not to be scheduled for transmission even if their CSI is known at the UAV.} 

Our objective is to maximize the minimum average data collection rate from all the SNs for the UAV in one single operation. Under the constraints on the UAV trajectory and communication scheduling, the optimization problem can be formulated as follows.
\begin{subequations}
\begin{align}
\max_{\mathbf{Q}, \mathbf{Z}, \mathbf{A},\eta} ~& \eta \nn \\  
\text{(P1)}\qquad\text{s.t.}~&\frac{1}{N}\sum_{n=1}^{N}  a_{k,n}  r_{k,n}\ge\eta, ~\quad\qquad\forall k\in \mathcal{K},\label{Eq:P1MinRateCons}\\
~&||{\mathbf q}_{n+1}-{\mathbf q}_n||\le S_{\rm{xy},\max},   ~~\quad \forall n \in\mathcal{N}, \label{Eq:P1ConsStar}\\
& |z_{n+1}-z_n|\le S_{\rm{z},\max}, \quad\qquad\forall n \in\mathcal{N}, \label{Eq:P1VerSpeed}\\
& ({\mathbf q}_1^T, z_1)=({\mathbf{q}}_I^{T}, z_I), ~~ ({\mathbf q}_{N+1}^T, z_{N+1})=({\mathbf{q}}_F^{T}, z_F),\\
& H_{\min}\le z_n\le H_{\max}, ~\quad\qquad\forall n \in\mathcal{N},\label{Eq:zmin}\\
&\sum_{k=1}^{K} a_{k,n}\le 1, ~~~\qquad\quad\qquad\forall n\in\mathcal{N}, \label{Eq:ScheCons}\\
&a_{k,n}\in\{0,1\},  ~~~\qquad\qquad\quad\forall k\in\mathcal{K}, n\in\mathcal{N}, \label{Eq:P1InterCons}
\end{align}
\end{subequations}
where $\mathbf{Q}=\{\mathbf{q}_n\}_{n=1}^{N+1}$, $\mathbf{Z}=\{z_n\}_{n=1}^{N+1}$, and $\mathbf{A}=\{a_{k,n}, \forall k\}_{n=1}^{N}$.

The optimal solution to problem (P1), in general, is difficult to obtain due to the lack of the \emph{complete} UAV-SNs CSI at all possible UAV locations in the 3D region of interest.\footnote{\color{black}{Note that even with complete UVA-SNs CSI, under our proposed probabilistic LoS channel model, problem (P1) is still intractable since the corresponding rate function is UAV location-dependent rather than a simple function w.r.t. the UAV-SN distance only, thus rendering the dynamic programming  or graph theory based approach proposed in \cite{zhang2018cellular} for trajectory optimization inapplicable. Thus, the optimal trajectory can only be obtained by an exhaustive search.}}
To address this difficulty, a key observation is that, in addition to the probabilistic LoS channel model which is known \emph{a priori}, the UAV can obtain the \emph{instantaneous} CSI with SNs in \emph{real time} along its flight, which allows the UAV to online adjust its trajectory and communication scheduling adaptive to the random  building blockage. Motivated by this, we propose in this paper a \emph{novel} and \emph{general} method to derive a {\color{black}suboptimal solution  to problem (P1)}, called \emph{hybrid offline-online optimization}, by leveraging both the statistical and real-time CSI. Our proposed method  consists of the following two phases as illustrated in Fig.~\ref{Fig:Hybrid}, which are briefly described as follows and will be elaborated in more details in the subsequent sections.
\begin{itemize}
\item[1)] \textbf{Offline phase}: Prior to the UAV's flight, we design an \emph{initial} UAV trajectory and communication scheduling policy  based on merely the probabilistic LoS channel model. The policy  is computed offline by solving an optimization problem to maximize the minimum \emph{expected} rate from all the SNs. The resultant 3D UAV trajectory yields a \emph{statistically favorable} UAV path that specifies the route (flying direction) the UAV follows along the trajectory, characterized by a sequence of ordered waypoints and line segments connecting them. 
\item[2)] \textbf{Online phase}:  In the online phase, we fix the UAV path (waypoints) as that obtained from the offline phase. Then, at each waypoint, without changing the  UAV flying direction, we formulate an LP to maximize the \emph{updated} minimum expected  rate from all the SNs via adjusting the UAV (horizontal and vertical) flying speeds for the remaining line segments of the offline optimized path  as well as its communication scheduling  with SNs over them.
\end{itemize} 
Note that different from the prior  joint design of UAV trajectory and communication scheduling (see e.g.,\cite{wu2018joint,zhan2018energy,you20193d}), our proposed method \emph{decouples} the design   into the UAV path optimization in the offline phase, followed by the real-time adjustment of  the UAV flying speeds and communication scheduling in the online phase. This  is motivated by the fact that the path optimization requires solving a time-consuming non-linear optimization problem (as will  be detailed in Section~\ref{Sec:Off}) which is desired to be implemented offline; while the online adaptation  can be designed  by solving an LP, which requires low  computational complexity and thus can be implemented at the UAV in real time. 

\begin{figure}[t]
\begin{center}
\includegraphics[height=6.2cm]{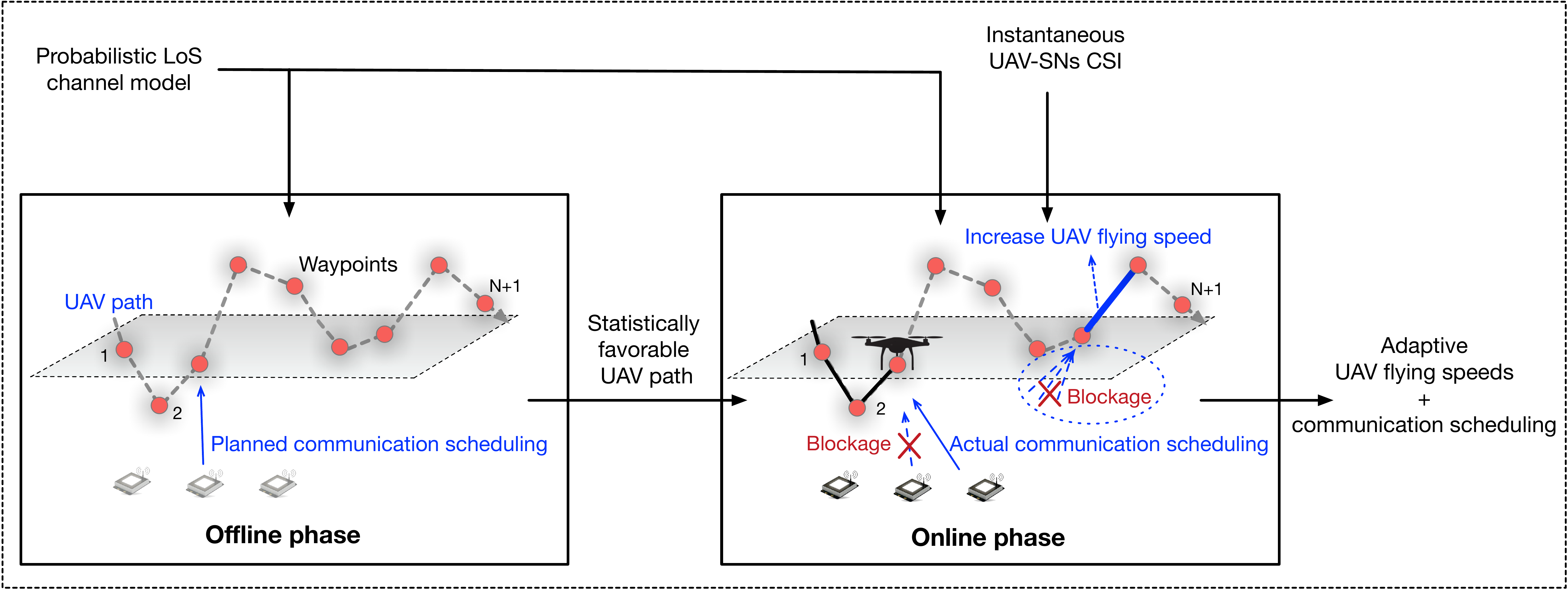}
\caption{The proposed hybrid offline-online optimization for UAV-enabled WSNs.}
\label{Fig:Hybrid}
\end{center}
\end{figure}
Despite low real-time computational complexity, our proposed hybrid design is expected to achieve superior  rate performance owing to the following reasons. First, the offline phase ensures  a statistically favorable UAV path that well balances the angle-distance trade-off for maximizing the average rates with SNs over the ensemble of city realizations. Second, as compared to the offline optimized policy which may result in considerable rate loss due to the random and unknown building blockage in the actual environment, our proposed online design, built upon the offline obtained path, endows the UAV with \emph{self-adaptation} to actual environment for further improving the rate performance. For instance, as illustrated in Fig.~\ref{Fig:Hybrid}, if an SN scheduled for data transmission at an arbitrary line segment according to the offline policy encounters building blockage in real time, the UAV can exploit the \emph{macro-diversity} to schedule another SN  that happens to be in the LoS  state for transmission. In a more challenging scenario, if all the SNs scheduled for transmission by the offline policy  are trapped in the NLoS states in real time, the UAV can fly over this line segment at the maximum  (horizontal and vertical) speeds  to save time for future transmission of SNs with better channel states. On the contrary, under highly   favorable channels with some SNs, the UAV can also slow down in  flying   to collect more data from them.

\section{Proposed Offline Design}\label{Sec:Off}
This section aims to offline design  the 3D UAV trajectory and communication scheduling for maximizing the minimum expected data collection rate from all the SNs under the probabilistic LoS channel model. The computed UAV path constituting a sequence of ordered waypoints and line segments will be utilized in the online design  in the next section.

\subsection{Problem Transformation}
In the offline phase, we focus on designing the initial 3D UAV trajectory and communication scheduling to achieve statistically favorable rate performance under the probabilistic LoS channel model. To this end, we first derive the expected rate from each SN, say SN $k$, in time slot $n$ as follows by combing \eqref{Eq:StateRate} and \eqref{Eq:LoSProb}.
\begin{equation}
\mathbb{E}[r_{k,n}]=P^{\rm{L}}_{k,n} r^{\rm{L}}_{k,n}+(1-P^{\rm{L}}_{k,n})r^{\rm{N}}_{k,n}, \quad \forall k, n.\label{Eq:ExpRate}
\end{equation}
Then problem (P1) is recast as follows with the aim  to maximize the minimum expected (average) rate from all the SNs, where the achievable rate $r_{k,n}$ is replaced by the expected rate $\mathbb{E}[r_{k,n}]$ given in \eqref{Eq:ExpRate}.
\begin{subequations}
\begin{align}
\max_{\mathbf{Q}, \mathbf{Z}, \mathbf{A}, \mathbf{\Theta}, \eta} ~& \eta \nn \\  
\text{(P2)}~~~~~~\text{s.t.}~&\frac{1}{N}\sum_{n=1}^N a_{k,n} \mathbb{E}[r_{k,n}]\ge\eta, \quad\forall k,\\
& \theta_{k,n}=\frac{180}{\pi}\arctan\l(\frac{z_n}{|| {\mathbf{q}}_n-{\mathbf{w}}_k||}\r), ~~\forall k,n, \label{Eq:OffAngleCons}\\
&\eqref{Eq:P1ConsStar}-\eqref{Eq:P1InterCons},\nn
\end{align}
\end{subequations}
where $\mathbf{\Theta}=\{\theta_{k,n}, \forall k\}_{n=1}^N$.

Problem (P2) is difficult to solve due to the non-concave rate function $\mathbb{E}[r_{k,n}]$ and the non-convex  binary scheduling constraints \eqref{Eq:P1InterCons}. To be specific, one can observe from \eqref{Eq:LoSProb}--\eqref{Eq:Distance} and \eqref{Eq:ExpRate} that $\mathbb{E}[r_{k,n}]$ is a highly complicated function of the 3D UAV  trajectory, due to  not only  the LoS/NLoS achievable rates but also the LoS probability. It is worth mentioning that in the existing works that consider  the probabilistic LoS channel model (e.g., \cite{esrafilian2018learning,zeng2019energy}), the expected rate is usually  approximated by (in contrast to that given in  \eqref{Eq:ExpRate})  
\begin{equation}
\mathbb{E}[r_{k,n}]\approx \bar{r}^{\rm h}_{k,n}\triangleq  \log_2\l(1+\dfrac{\mathbb{E}[{h}_{k,n}] P_k}{\sigma^2 \Gamma}\r),\label{Eq:TradRateAppr}
\end{equation}
where $\mathbb{E}[{h}_{k,n}]=P^{\rm{L}}_{k,n} h_{k,n}^{\rm{L}}+(1-P^{\rm{L}}_{k,n})h^{\rm{N}}_{k,n}$. {\color{black}Such an approximation  cannot guarantee the rate performance since the resultant  approximate optimization problem indeed  maximizes an upper bound of the expected rate due to Jensen's inequality and the gap is non-negligible because of the significant disparity between the channel strengths under the LoS and NLoS states.}
To achieve more accurate approximation, 
an important observation is that, given the UAV's location, the rate in the LoS state is practically much larger than that in the NLoS state due to the additional signal attenuation $\mu$ and a larger  path loss exponent $\alpha_{\rm N}$ (see \eqref{Eq:BinaryChannel}). This  implies that we can   lower-bound the expected rate function in \eqref{Eq:ExpRate} as below that only accounts for the expected rate in the LoS state and thus is achievable, i.e.,
\begin{align}
\!\!\!\mathbb{E}[{r_{k,n}}]&\ge P^{\rm{L}}_{k,n} r^{\rm{L}}_{k,n}\nn\\
&=\l(B_3+\frac{B_4}{1+ e^{-(B_1+B_2 \theta_{k,n})}}\r)\times \log_2\l(1+\dfrac{\gamma_k}{(|| {\mathbf{q}}_n-{\mathbf{w}}_k||^2+z_n^2)^{\alpha_{\rm{L}}/2}}\r)\overset{\triangle}{=}\bar{r}^{\rm{L}}_{k,n}.  \label{Eq:ApprExpRate}
\end{align}
{\color{black} An illustrative example is provided as follows for demonstrating the improved accuracy and achievability of the proposed expected-rate approximation. The system parameters are given by:  $\mu=20$ dB, $\alpha_{\rm L}=2.5$, $\alpha_{\rm N}=3.5$, $\beta_0=-60$ dB, and $\gamma_k=60$ dB. In time slot $n$ and for a typical user $k$ with $d_{k,n}=50$ m, 
the corresponding UAV-SN channel power gain  conditioned on the LoS and NLoS states are obtained as $h_{k,n}^{\rm L}=-102.5$ dB and $h_{k,n}^{\rm N}=-139.5$ dB, respectively. As such, the achievable rates in the LoS and NLoS states are given by $r_{k,n}^{\rm L}=5.85$ bps/Hz and $r_{k,n}^{\rm N}=0.016$ bps/Hz, respectively. For a typical LoS probability, e.g., $P^{\rm L}_{k,n}=0.5$, the actual expected rate is $\mathbb{E}[{r_{k,n}}]=2.93$ bps/Hz. Thus, the newly-approximated expected rate, obtained as  $\bar{r}_{k,n}^{\rm L}=2.92$ bps/Hz,  is achievable and   much more accurate than the conventionally overestimated expected rate  $\bar{r}^{\rm h}_{k,n}=4.87$ bps/Hz.} 
Based on the above expected-rate lower bound, problem (P2) can be reformulated into the following approximate form.
\vspace{-5pt}
\begin{subequations}
\begin{align}
\max_{\mathbf{Q}, \mathbf{Z}, \mathbf{A}, \mathbf{\Theta}, \eta }  &~\eta \nn \\  
\text {(P3)}~~~~~\text{s.t.}~~&\frac{1}{N}\sum_{n=1}^N a_{k,n} \bar{r}^{\rm{L}}_{k,n}\ge\eta, \quad\forall k, \label{Eq:OffMinRateCons}\\
& \eqref{Eq:P1ConsStar}-\eqref{Eq:P1InterCons}, \eqref{Eq:OffAngleCons}.\nn
\end{align}
\end{subequations}

\subsection{Proposed Algorithm for Problem (P3)}

Problem (P3) is still challenging to solve due to the coupled horizontal and vertical trajectory variables in the non-convex rate constraints \eqref{Eq:OffMinRateCons} and the non-affine elevation-angle constraints \eqref{Eq:OffAngleCons}, as well as the  integer variables in the transmission scheduling constraints \eqref{Eq:P1InterCons}.  To tackle these difficulties, we first relax the integer constraints for the communication scheduling, leading to the following relaxed  problem
\begin{subequations}
\begin{align}
\max_{\mathbf{Q}, \mathbf{Z}, \mathbf{A}, \mathbf{\Theta}, \eta } &~~\eta \nn \\  
\text{(P4)}\qquad \text{s.t.}~~&0\le a_{k,n}\le1, \quad \forall k, n,\label{Eq:RelaxScheduling}\\
& \eqref{Eq:P1ConsStar}-\eqref{Eq:ScheCons}, \eqref{Eq:OffAngleCons}, \eqref{Eq:OffMinRateCons}.\nn
\end{align}
\end{subequations}
Then, to address the non-affine constraints \eqref{Eq:OffAngleCons}, 
we can prove by contradiction that the optimal solution to problem (P4) is the same as that to the following further relaxed problem
\begin{subequations}
\begin{align}
\max_{\mathbf{Q}, \mathbf{Z}, \mathbf{A}, \mathbf{\Theta}, \eta } &~~\eta \nn \\  
\text{(P5)}\qquad \text{s.t.}~~
& \theta_{k,n}\le \frac{180}{\pi}\arctan\l(\frac{z_n}{|| {\mathbf{q}}_n-{\mathbf{w}}_k||}\r), \forall k,n,\label{Eq:RalxTheta}\\
& \eqref{Eq:P1ConsStar}-\eqref{Eq:ScheCons},  \eqref{Eq:OffMinRateCons},\eqref{Eq:RelaxScheduling}.\nn
\end{align}
\end{subequations}
However, problem (P5) is still non-convex for which the optimal solution is hard to obtain. As such, we  propose an efficient iterative  algorithm in the following  based on  BCD to obtain  a suboptimal  solution to it. 

\subsubsection{Communication Scheduling Optimization}
Given any feasible 3D  UAV trajectory $\{\mathbf{Q}, \mathbf{Z}\}$, problem (P5) can be rewritten as the following problem
\vspace{-5pt}
\begin{subequations}
\begin{align}
\max_{\mathbf{A}, \eta } ~~&\eta \nn \\  
\text{(P6)}\qquad\text{s.t.}~~&\eqref{Eq:ScheCons}, \eqref{Eq:OffMinRateCons},
\eqref{Eq:RelaxScheduling}.\nn
\end{align}
\end{subequations}
Problem (P6) is a standard LP which  can be efficiently solved by existing solvers, e.g., CVX \cite{CVX}. Note that the continuous communication scheduling obtained from solving problem (P6) can be reconstructed to the binary scheduling using the method in \cite{wu2018joint} without compromising the optimality.

\subsubsection{UAV Horizontal Trajectory Optimization}
Given any feasible communication scheduling, $\mathbf{A}$, and UAV vertical trajectory, $\mathbf{Z}$, problem (P5)  reduces to the  problem below for optimizing the UAV horizontal trajectory.
\begin{subequations}
\begin{align}
\max_{\mathbf{Q}, \mathbf{\Theta},\eta }  ~~&\eta \nn \\  
\text{(P7)}\qquad\text{s.t.}~~&{\mathbf q}_1={\mathbf{q}}_I, \quad {\mathbf q}_{N+1}={\mathbf{q}}_F, \label{Eq:P8StarEnd} \\
&\eqref{Eq:P1ConsStar},\eqref{Eq:OffMinRateCons}, \eqref{Eq:RalxTheta}.\nn
\end{align}
\end{subequations}
To solve this non-convex optimization problem, we first introduce an important lemma as follows.

\begin{lemma}\label{Lem:ConvexRateFunc}\emph{Given $\gamma\ge0$ and $\alpha\ge2$,  $\psi(x,y)\overset{\triangle}{=}\(B_3+\frac{B_4}{x}\r)\log_2\l(1+\frac{\gamma}{y^{\alpha/2}}\r)$ is a convex function for $x>0$ and $y>0$.}
\end{lemma}
\begin{proof}
See Appendix~\ref{App:ConvexRateFunc}.
\end{proof}

Using Lemma~\ref{Lem:ConvexRateFunc}, we can prove that $\bar{r}^{\rm{L}}_{k,n}$ in \eqref{Eq:ApprExpRate} is a convex function w.r.t. $(1+ e^{-{B_1+B_2 \theta_{k,n}}})$ and $(|| \mathbf{q}_n-\mathbf{w}_k||^2+z_n^2)$. As such, we can apply the SCA technique to approximate the rate function $\bar{r}^{\rm{L}}_{k,n}$ by its lower bound as follows using the first-order Taylor expansion.{\footnote{\color{black}The proposed offline design for the rotary-wing UAV can be extended to the fixed-wing UAV by adding the constraints of the minimum  UAV  horizontal and vertical flying speeds and handling these non-convex constraints by applying SCA techniques.}

\begin{lemma}\label{Lem:RateBound}\emph{
For any local UAV horizontal trajectory $\hat{\mathbf{Q}}$, $\bar{r}^{\rm{L}}_{k,n}$ given in \eqref{Eq:ApprExpRate} can be lower-bounded by
\vspace{-10pt}
\begin{align}\label{Eq:RTylor}
\bar{r}^{\rm{L}}_{k,n}&\ge \hat{\bar{r}}^{\rm{L}}_{k,n}-\hat{\Omega}_{k,n} (e^{-\phi_{k,n}}-e^{-\hat{\phi}_{k,n}})-\hat{\Psi}_{k,n}(|| \mathbf{q}_n-\mathbf{w}_k||^2-|| \hat{\mathbf{q}}_n-\mathbf{w}_k||^2) \nn\\
&\overset{\triangle}{=}\hat{\bar{r}}^{\rm{L},  \rm{lb}}_{k,n}, \qquad \forall k,n,
\end{align}
where $\phi_{k,n}=B_1+B_2\theta_{k,n}$, and the coefficients $\hat{\bar{r}}^{\rm{L}}_{k,n}$, $\hat{\Omega}_{k,n}$, $\hat{\Psi}_{k,n}$ and $\hat{\phi}_{k,n}$ are defined in Appendix~\ref{App:RateBound}. 
The equality holds at the point $\mathbf{q}_n=\hat{\mathbf{q}}_n$.}
\end{lemma}
\begin{proof}
See Appendix~\ref{App:RateBound}.
\end{proof}
For the  non-convex constraints \eqref{Eq:RalxTheta}, let 
\begin{equation}
v_{k,n}\overset{\triangle}{=}\arctan\l(\frac{z_n}{|| {\mathbf{q}}_n-{\mathbf{w}}_k||}\r).\label{Eq:v}
\end{equation} One can observe that, although $v_{k,n}$ is not a concave function w.r.t. ${\mathbf{q}}_n$, it is a convex function w.r.t. $(|| {\mathbf{q}}_n-{\mathbf{w}}_k||)$ since $\arctan(1/x)$ is convex for $x> 0$.
This useful property allows us to lower-bound $v_{k,n}$ as follows by using the SCA technique.

\begin{lemma}\label{Lem:vSCA}\emph{
For any local UAV horizontal trajectory $\hat{\mathbf{Q}}$, $v_{k,n}$ given  in \eqref{Eq:v} can be lower-bounded by
\vspace{-10pt}
\begin{align}
v_{k,n}\ge  \hat{v}_{k,n}-\hat{\Lambda}_{k,n}(|| {\mathbf{q}}_n-{\mathbf{w}}_k||-|| {\hat{\mathbf{q}}}_n-{\mathbf{w}}_k||)\overset{\triangle}{=} \hat{v}^{\rm{lb}}_{k, n},\qquad\forall k,n, 
\end{align}
where  the coefficients $ \hat{v}_{k,n}$ and $\hat{\Lambda}_{k,n}$ are respectively given by 
\begin{align}
 \hat{v}_{k,n}=\arctan\l(\frac{z_n}{|| {\hat{\mathbf{q}}}_n-{\mathbf{w}}_k||}\r),~~~\hat{\Lambda}_{k,n}=\frac{z_n}{|| {\hat{\mathbf{q}}}_n-{\mathbf{w}}_k||^2+z_n^2}.
\end{align}
The equality holds at the point $\mathbf{q}_n=\hat{\mathbf{q}}_n$.
}
\end{lemma}

Using Lemmas~\ref{Lem:RateBound} and \ref{Lem:vSCA},  problem (P7) can be reformulated into the approximate form given  below, where $\bar{r}^{\rm{L}}_{k,n}$ in \eqref{Eq:OffMinRateCons} and $v_{k,n}$ in \eqref{Eq:RalxTheta} are replaced by their corresponding lower bounds.
\begin{subequations}
\begin{align}
\max_{\mathbf{Q}, \mathbf{\Phi}, \mathbf{\Theta},\eta} ~&\eta\nn \\  
\text{(P8)}\qquad\text{s.t.}~~&\frac{1}{N}\sum_{n=1}^{N} a_{k,n}  \hat{\bar{r}}^{\rm{L},  \rm{lb}}_{k,n}\ge \eta, \qquad\forall k,\label{Eq:P9Rbound}\\
\qquad\qquad &\phi_{k,n}=B_1+B_2\theta_{k,n}, \qquad \forall k,n, \label{Eq:utheta}\\
& \theta_{k,n}\le \frac{180}{\pi}   \hat{v}^{\rm{lb}}_{k, n}, \qquad\qquad\forall k,n,\label{Eq:P9Equality}\\
& \eqref{Eq:P1ConsStar}, \eqref{Eq:P8StarEnd},\nn
\end{align}
\end{subequations}
where $\mathbf{\Phi}=\{\phi_{k,n}, \forall k\}_{n=1}^{N}$.
The approximate problem  (P8) is a convex optimization problem and thus can be efficiently solved by using existing solvers, e.g., CVX. It is worth mentioning that, by approximating the non-convex constraints with their convex lower bounds, the feasible set of  problem (P8) is always a subset of problem (P7). This guarantees that solving problem (P8) gives a lower bound of the optimal objective value of problem (P7).

\subsubsection{UAV Vertical Trajectory Optimization}
Last, given any feasible communication scheduling, $\mathbf{A}$, and UAV horizontal trajectory, $\mathbf{Q}$, problem (P5)  reduces to the following problem for optimizing the UAV vertical trajectory.
\begin{subequations}
\begin{align}
\max_{\mathbf{Z},  \mathbf{\Theta}, \eta} ~~~&\eta  \nn\\  
\text{(P9)}\qquad \text{s.t.}\quad& z_1=z_I, ~~~z_{N+1}=z_F,\label{Eq:P11z}\\
&\eqref{Eq:P1VerSpeed}, \eqref{Eq:zmin}, \eqref{Eq:OffMinRateCons}, \eqref{Eq:RalxTheta}.\nn
\end{align}
\end{subequations}
Since problem (P9) has a similar form as problem (P7), we can apply a similar approach as for solving problem (P7) (i.e., applying the SCA technique to the constraints \eqref{Eq:OffMinRateCons}) and thereby  problem (P9) can be transformed to the following problem
\vspace{-5pt}
\begin{subequations}
\begin{align}
\max_{\mathbf{Z}, \mathbf{\Phi},  \mathbf{\Theta}, \eta} ~~&\eta  \nn\\ 
\text{(P10)}\qquad\text{s.t.}~~&\frac{1}{N}\sum_{n=1}^{N} a_{k,n} \check{\bar{r}}^{\rm{L},\rm{lb}}_{k,n}\ge \eta, \forall k, \label{Eq:P11RateCons}\\
& \eqref{Eq:P1VerSpeed}, \eqref{Eq:zmin},  \eqref{Eq:RalxTheta},\eqref{Eq:utheta}, \eqref{Eq:P11z},\nn
\end{align}
\end{subequations}
where $\check{\bar{r}}^{\rm{L},\rm{lb}}_{k,n}$ is the lower bound of $\bar{r}^{\rm{L}}_{k,n}$ given a local vertical trajectory $\check{\mathbf{Z}}$, which can be obtained by using a similar SCA technique as in Lemma~\ref{Lem:RateBound}. Note that in \eqref{Eq:RalxTheta}, $v_{k,n}$ is a convex function w.r.t. $z_n$, since $\arctan(x)$ is a concave function for $x>0$.  Thus, problem (P10) is a convex optimization problem, which can be optimally solved by using the  interior-point method in general. In practice, due to the lack of support for the function of $\arctan(x)$ in CVX, we can approximate  $\arctan(x)$ by its upper bound using the first-order Taylor expansion for simplicity.

\subsubsection{Overall Algorithm and Computational Complexity}\label{Sec:Complexity}

  
Using the results in the preceding subsections, an iterative algorithm can be proposed to 
 obtain a suboptimal solution to problem (P3)  by  alternately  optimizing the transmission scheduling of SNs, UAV horizontal  and vertical trajectories. In addition,  the initial UAV trajectory can be constructed as the straight flight from the initial location to the final location.  {\color{black}Next,  we discuss  the  computational complexity of the proposed  algorithm. Since the communication scheduling, and UAV horizontal and vertical trajectories are sequentially optimized in each iteration by using the CVX solver which is based on the standard interior-point method,  their individual  complexities scale as  $\mathcal{O}((KN)^{3.5}\log(1/\epsilon))$, $\mathcal{O}((N+KN)^{3.5}\log(1/\epsilon))$, and $\mathcal{O}((N+KN)^{3.5}\log(1/\epsilon))$,  respectively, given a solution accuracy of $\epsilon> 0$ \cite{ben2001lectures}. Then accounting for the BCD iterations with the complexity of $\log(1/\epsilon)$, the total computational complexity of the algorithm is thus  $\mathcal{O}((N+KN)^{3.5}\log^2(1/\epsilon))$, which is affordable for offline computation. Moreover, the algorithm is shown to always converge in our extensive simulations (see Section~\ref{Sec:SimOff}).} Last, the UAV path obtained from the offline phase is represented by the waypoints $\{({\mathbf{q}}_n^{T*}, z^*_n)\}_{n=1}^{N+1}$, which denote a suboptimal solution to problem (P3) computed by the iterative algorithm.
\vspace{-10pt}
\section{Proposed Online Design}\label{Sec:On}
Directly implementing the offline optimized UAV trajectory and communication scheduling policy in practice may suffer considerable rate loss due to the lack of online adaptation to the random  building blockage. This issue is addressed in this section by designing the online policy that adaptively adjusts the UAV flying speeds along the offline optimized UAV path as well as its communication scheduling with SNs based on the real-time UAV-SNs CSI  and SNs' individual amounts of data  received accumulatively.

To this end, the UAV path (instead of time) is discretized into $N$ line segments with the $(N+1)$ waypoints fixed as those obtained from the offline phase, i.e., $\{({\mathbf{q}}_n^{T*}, z^*_n)\}_{n=1}^{N+1}$ \cite{zeng2019energy}. At each line segment, the distance between the UAV and each  SN can be assumed to be approximately unchanged due to the  time discretization method adopted  in the offline phase. Unless otherwise stated, we reuse the notations w.r.t. the time slot under time discretization for  the current case w.r.t. the line segment for convenience. In the online phase, we assume that the UAV flies at a constant speed over each line segment horizontally as well as  vertically.  To enforce that the UAV follows the offline optimized path, the time durations spent  on traveling the horizontal and vertical distances at each line segment $n$ need to be identical and thus commonly denoted by $t_n, \forall n$. Then the UAV flying velocity at each line segment is given by 
\begin{equation} 
\mathbf{v}_n=\frac{\(\mathbf{q}^{T*}_{n+1}-\mathbf{q}^{T*}_{n}, z^*_{n+1}-z^*_n\r)^{T}}{t_n}, \quad \forall n\in\mathcal{N},\label{Eq:velocity}
\end{equation}
and the corresponding horizontal and vertical flying speeds are respectively  given by 
\begin{equation}
v_{{\rm{xy}},n}=\frac{||\mathbf{q}^*_{n+1}-\mathbf{q}^*_{n} ||}{t_n},~~~v_{{\rm{z}},n}=\frac{|z^*_{n+1}-z^*_n |}{t_n}, \quad \forall n\in\mathcal{N}.\label{Eq:flyspeed}
\end{equation}
As a result, the optimization of the UAV flying speeds at each line segment can be equivalently transformed to that of the traveling duration. For the proposed online policy, let  $ T^{\rm{re}}_n$ denote the remaining  flight duration from  the beginning of each line segment $n$, which evolves as
\vspace{-5pt}
\begin{equation}
T^{\rm{re}}_1=T_0,~~  T^{\rm{re}}_n= T^{\rm{re}}_{n-1}-t_{n-1}, n>1.
\end{equation}
Moreover, instead of using the communication scheduling  $a_{k,n}$ obtained  in the offline phase, we define $\tau_{k,n}\ge 0$ as the allocated data-transmission time for SN $k$ at line segment $n$. Then the amount of received data  (in bits/Hz) at the UAV from SN $k$ over the $n$-th line segment is given by $\tau_{k,n} r_{k,n}$ and the amount of  accumulatively received data up to  the beginning of each line segment $n$, denoted by $r^{\rm{ac}}_{k,n}$, evolves as  
\vspace{-5pt}
\begin{equation}
r^{\rm{ac}}_{k,1}=0, ~~~r^{\rm{ac}}_{k,n}=\sum_{m=1}^{n-1}  \tau_{k,m}  r_{k,m},~~ n>1, \forall k\in\mathcal{K}.
\end{equation}

To formulate the optimization problem in the online phase, since  the UAV re-optimizes its  policy at each waypoint $n$ including its traveling durations and communication durations with SNs over all the subsequent line segments $m=n, \cdots, N$, we re-denote $\tau_{k,m}$ and $t_{k,m}$ by $\tau_{k,m}^{(n)}$ and $t_m^{(n)}$ as the variables in the $n$-th policy  optimization. The online  policy needs to satisfy the following constraints. First, the finite UAV horizontal and vertical flying speeds enforce that 
\begin{equation}
v_{{\rm{xy}},m}\le V_{\rm{xy},\max},~~ v_{{\rm{z}},m}\le V_{\rm{z},\max},\quad m=n, \cdots, N.\footnote{\color{black}{The proposed  online design for the rotary-wing UAV can be readily extended to the fixed-wing UAV by adding the similar constraints of the  minimum UAV  horizontal and vertical flying speeds and converting them into the equivalent constraint on the flying duration over each line segment.}}\label{Eq:pathspeed}
\end{equation}
Combining  \eqref{Eq:pathspeed} and \eqref{Eq:flyspeed} yields
\begin{equation}
t_m^{(n)}\ge\max\l\{\frac{||\mathbf{q}^{*}_{m+1}-\mathbf{q}^*_{m} ||}{V_{\rm{xy},\max}}, \frac{|z^*_{m+1}-z^*_m |}{V_{\rm{z},\max}}\r\}\overset{\triangle}{=}\hat{t}_m, \quad m=n, \cdots, N.\label{Eq:OnlineEachT}
\end{equation}
Second, the total  traveling  duration over  the remaining line segments should satisfy $\sum_{m=n}^{N} t_m^{(n)} \le  T^{\rm{re}}_n$.
Last, the communication scheduling constraint in \eqref{Eq:ScheCons} is re-written by 
\begin{equation}
\sum_{k=1}^K \tau_{k,m}^{(n)} \le t_m^{(n)},\quad  m=n, \cdots, N.\label{Eq:OnlineTimeShare}
\end{equation} 
For the rate performance, let $r_{k,n}^{\rm{ave}}$ denote the \emph{updated} expected average rate from SN $k$ at waypoint $n$, which is given by
\begin{equation}
r_{k,n}^{\rm{ave}}=\frac{1}{T_0}\(r^{\rm{ac}}_{k,n}+\tau_{k,n}^{(n)} r_{k,n}+ \sum_{m=n+1}^{N} \tau_{k,m}^{(n)} \mathbb{E}[r_{k,m}]\),~~\forall k\in\mathcal{K},\label{Eq:UpAve}
\end{equation}
accounting for the amounts of data received accumulatively, $r^{\rm{ac}}_{k,n}$, achieved in the current line segment, $\tau_{k,n}^{(n)} r_{k,n}$, and  expected to be received over subsequent line segments, $\sum_{m=n+1}^{N} \tau_{k,m}^{(n)} \mathbb{E}[r_{k,m}]$. Note that in \eqref{Eq:UpAve}, $r_{k,n}$ is determined by the instantaneous CSI (see \eqref{Eq:StateRate}) and $\{\mathbb{E}[r_{k,m}],\forall k\}_{m=n+1}^N$ can be explicitly obtained from \eqref{Eq:ExpRate} with the LoS probabilities determined by the fixed waypoints according to \eqref{Eq:angle}--\eqref{Eq:LoSProb}. 
Based on the above discussion, the optimization problem for the online adaptation at each waypoint can be modified from that for the offline design, i.e., problem (P2), as formulated below.
\vspace{-5pt}
\begin{subequations}
\begin{align}
\max_{  \boldsymbol{\tau}^{(n)}, \mathbf{T}^{(n)},  \eta^{(n)}} ~~&\eta^{(n)}  \nn\\ 
\text{(P11)}~\qquad\text{s.t.}~~~~~~~&
r_{k,n}^{\rm{ave}}\ge \eta^{(n)},\quad \forall k\in\mathcal{K}, \label{Eq:P12RateCons}\\
&\tau_{k,m}^{(n)}\ge0,  \quad\quad\forall k\in \mathcal{K}, m=n, \cdots, N,\\
& \eqref{Eq:OnlineEachT}-\eqref{Eq:OnlineTimeShare},\nn
\end{align}
\end{subequations}
where $\boldsymbol{\tau}^{(n)}=\{\tau_{k,m}^{(n)}, \forall k\}_{m=n}^N$ and  $\mathbf{T}^{(n)}=\{ t_m^{(n)}\}_{m=n}^N$.

Problem (P11) is a standard LP and thus can be efficiently solved by existing solvers, e.g., CVX. Let $\{\tau_{k,n}^{(n)*}, \forall k\}_{m=n}^N$, $\{t^{(n)*}_m\}_{m=n}^N$, and $\eta^{(n)*}$ denote the optimal solution to problem (P11). By using contradiction, we can easily show  that the UAV tends to schedule  SNs for data transmission at the line segments with relatively high instantaneous achievable rate $r_{k,n}$ or the expected rate $\mathbb{E}[r_{k,m}]$, while at the same time, balancing with the SNs'  individual average rates. This observation indicates that the online policy  allows the UAV to adaptively  schedule data transmissions from  the SNs with high real-time achievable rates by exploiting the channel  macro-diversity. Moreover, for the line segment where most of  the SNs are in the NLoS states with relatively low achievable rates, the UAV is expected to reduce the traveling duration over it  to save time for future transmission of SNs with better channel states.

{\color{black}It is worth mentioning that, at each waypoint $n$, the  computational complexity for solving problem (P11) using the interior method  scales as $\mathcal{O}\l(((K+1)(N-n+1))^{3.5}\r)$, which decreases along the UAV path due to the decreasing number of the optimization variables associated with the remaining line segments. In practice, solving an LP takes relatively short running time as will be  demonstrated  in  Section~\ref{Sec:SimuOnline}. Last, combining the analysis  for the convergence and complexity of the offline design in Section~\ref{Sec:Complexity}, the proposed hybrid offline-online optimization method can be shown to always converge to a suboptimal solution of problem (P1) and has an overall complexity order of $\mathcal{O}\l(((K+1)(N-n+1))^{3.5}+(N+KN)^{3.5}\log^2(1/\epsilon)\r)$.}
\begin{remark}[Practical implementation]\emph{Prior to the UAV's flight, the initial joint design of UAV trajectory and communication scheduling is offline computed by the iterative algorithm in Section~\ref{Sec:Complexity}.
The information of the optimized waypoints, $\{({\mathbf{q}}_n^{T*}, z^*_n)\}_{n=1}^{N+1}$, is  stored at the UAV, based on which the UAV can compute the expected rates from individual SNs at each line segment, $\{\mathbb{E}[r_{k,n}], \forall k\}_{n=1}^N$. During the UAV's flight, at each waypoint $n$, the UAV first acquires the instantaneous CSI  with  SNs $\{c_{k,n}, \forall k\}$ and then computes the optimal solution to problem (P11) with the updated information of $\{r^{\rm{ac}}_{k,n}, \forall k\}$ and  $ T^{\rm{re}}_n$. Last, the UAV applies the communication scheduling for the current line segment $n$ according to $\{\tau_{k,n}^{(n)*}, \forall k\}$, and at the same time, heads towards the next waypoint at the horizontal and vertical speeds determined by $t^{(n)*}_n$ and \eqref{Eq:flyspeed}. 
}
\end{remark}

\vspace{-15pt}
\begin{remark}[DoF for online speed optimization]\emph{The total  flight duration $T_0$ affects the DoF for speed optimization in the online phase. Specifically,  given a relatively short flight duration, the UAV designed in the offline phase tends to sequentially visit  each SN as close as possible  at the maximum horizontal speed. As observed from \eqref{Eq:OnlineEachT}, this will limit the DoF for speed optimization in the online phase even with a large maximum vertical speed. In contrast, with a relatively long flight duration, the UAV can more flexibly  accelerate or slow down in its real-time flight as long as it can reach the final location in time.}
\end{remark}

\begin{remark}[Proposed method versus open-loop feedback control]\emph{
The  proposed online adaptation partially resembles the celebrated \emph{open-loop feedback control} (OLFC) \cite{bertsekas1995dynamic} that computes an open-loop policy at the current time as if no future (channel) state information will be available. However, when applying to our considered problem, OLFC requires solving a joint optimization problem for the UAV trajectory and communication scheduling in each time slot using an  algorithm similar to the iterative algorithm in Section~\ref{Sec:Complexity}
 with sequentially reduced dimension,  whose computational complexity may be too high to be implementable at the UAV in real time. In contrast, our proposed hybrid design solves the time-consuming path optimization problem in the offline phase, and only refines a \emph{simplified} open-loop policy in the online phase for adjusting  the UAV flying speeds and communication scheduling  by solving a low-complexity  LP, thus making it amenable to real-time  implementation by the UAV.
}
\end{remark}
\vspace{-7pt}
{\color{black}\begin{remark}[Extension to the energy-constrained SNs]\emph{The results of this paper for the renewable energy-powewed SNs can be extended to the energy-constrained SNs as follows. First, let $P_{k,n}$ denote the transmit power of SN $k$ in time slot $n$, which needs to satisfy the long-term power constraint: $\frac{1}{N}\sum_{n=1}^N P_{k,n}\le P_{k,\max}, \forall k$, where $P_{k,\max}$ is the maximum average transmit power of SN $k$. Then the approximated expected rate, $\bar{r}^{\rm{L}}_{k,n}$ in \eqref{Eq:ApprExpRate}, can be re-expressed as $\bar{r}^{\rm{L}}_{k,n}\triangleq\l(B_3+\frac{B_4}{1+ e^{-(B_1+B_2 \theta_{k,n})}}\r)\times \log_2\l(1+\dfrac{P_{k,n}\gamma_0}{a_{k,n}(|| {\mathbf{q}}_n-{\mathbf{w}}_k||^2+z_n^2)^{\alpha_{\rm{L}}/2}}\r)$. As such, the offline design can be modified by reformulating the communication-scheduling problem (P3) as the joint optimization problem for the communication scheduling and SNs' transmit power allocation, which can be shown to be a convex optimization problem. Similarly, the online design can be modified by reformulating the optimization problem (P11) as the joint optimization problem for the UAV's flying speeds as well as SNs' transmit power allocation, without changing the main solution approach proposed in this paper.}
\end{remark}}
\vspace{-15pt}
\section{Simulation Results}\label{Sec:Simu}

In this section, we provide extensive simulation results to verify the effectiveness of the proposed hybrid offline-online design and show key properties of the optimized 3D offline and online UAV trajectories as well as the adaptive communication scheduling.
Unless otherwise stated, the following results are based on a random  realization of $4$ SNs' locations over a square area of $300\times300$ $\rm{m}^2$ for ease of illustration. The UAV is assumed to fly from an initial location  $(0,150,50)$ m to a final location $(300, 150, 50)$ m.  Assume that all the SNs have the same transmit power of $0.1$ W.  We consider an urban environment  for which the parameters of its corresponding generalized logistic model for the LoS probability  are given by $B_1=-0.4568$, $B_2=0.0470$, $B_3=-0.63$, and $B_4=1.63$. Other parameters are set as $\beta_0 = -60$ dB, $\Gamma = 8.2$ dB, $\sigma^2 = -109$ dBm, $\alpha_{\rm{L}}=2.5$, $\alpha_{\rm{N}}=3.5$, $\mu=-20$ dB, $V_{\rm{xy},\max} = 40$ m/s, $V_{\rm{z},\max} = 20$ m/s, $H_{\min} = 50$ m, $H_{\max} = 300$ m, $\delta=0.2$ s, and $\epsilon=0.001$.

\begin{figure}[t]
\begin{center}
\includegraphics[height=5.2cm]{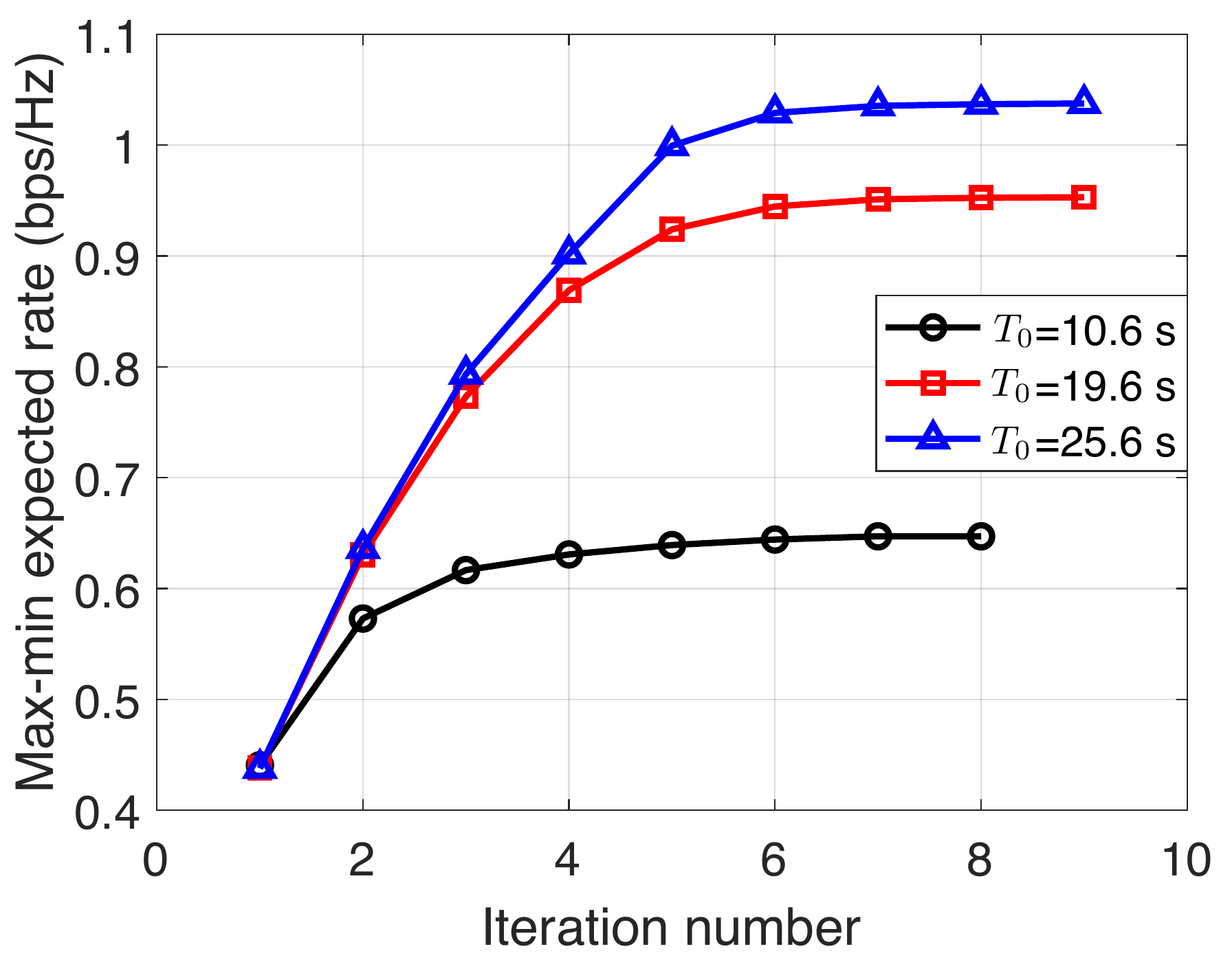}
\caption{Max-min expected rate versus iteration number.}
\label{Fig:Conv}
\end{center}
\end{figure}
\vspace{-5pt}
\subsection{Offline Phase}\label{Sec:SimOff}
We first evaluate the performance of the proposed initial 3D UAV  trajectory design in the offline phase.  To start with, we first show the convergence behavior of the iterative algorithm in Section~\ref{Sec:Complexity}
under different flight durations in Fig.~\ref{Fig:Conv}. It can be observed that the max-min expected rate based on the probabilistic LoS channel model monotonically increases with the number of iterations and quickly converges after  around $10$ iterations for all flight durations.

For  performance comparison, we consider two benchmark schemes: 1) LoS-based (LB) scheme, which jointly designs UAV trajectory and communication scheduling based on the simplified LoS channel model (i.e., assuming the LoS probability given in \eqref{Eq:LoSProb} to be $P^{\rm{L}}_{k,n}=1, \forall k, n$); 2) probabilistic-LoS with the lowest altitude (PLLA) that only optimizes the UAV horizontal trajectory and communication scheduling with the UAV altitude kept equal to $H_{\min}$ for all time. 
Our proposed scheme is named as probabilistic-LoS based (PLB) scheme for convenience. Fig.~\ref{Fig:Traj106} illustrates the UAV trajectories by different schemes with the flight duration $T_0=10.6$ s. Several interesting observations are made as follows. First, as shown in Fig.~\ref{FigHoritraj_sch}, the three schemes have similar horizontal trajectories with the UAV sequentially traveling around each SN, but differ in that, the UAV for both the PLB and PLLA schemes moves closer to SNs $1$ and  $2$ than that of the LB scheme when traveling nearby them, at the cost of being more away from SN $4$. The reason is that the elevation angles of the UAV with SNs $1$ and $2$ are relatively small and thus need to  be enlarged by tuning the UAV horizontal trajectory, which is not seen in the trajectory by the LB scheme. Second, for the 3D UAV trajectory shown in Fig.~\ref{Fig3DTraj_sch}, we can observe that the proposed PLB scheme can exploit the additional  DoF brought by  the UAV vertical trajectory to balance the angle-distance trade-off more efficiently than the PLA scheme as the elevation angle can be effectively enlarged by moderately increasing the altitude  without incurring significant path loss.

\begin{figure}[t!]
\centering
\subfigure[UAV horizontal trajectory.]{\label{FigHoritraj_sch}
\includegraphics[height=5.3cm]{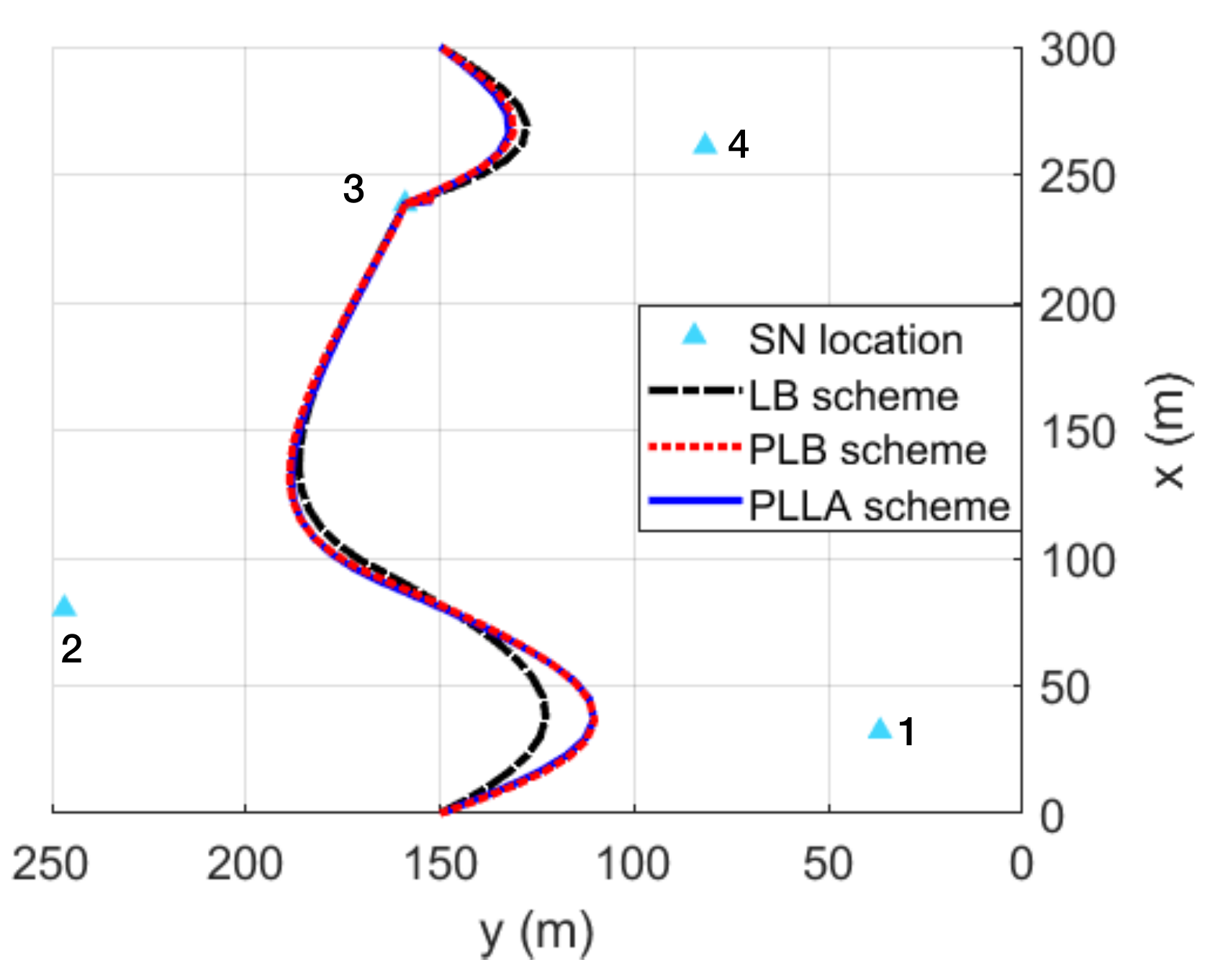}}
\hspace{10mm}
\subfigure[3D UAV  trajectory.]{\label{Fig3DTraj_sch}
\includegraphics[height=5.3cm]{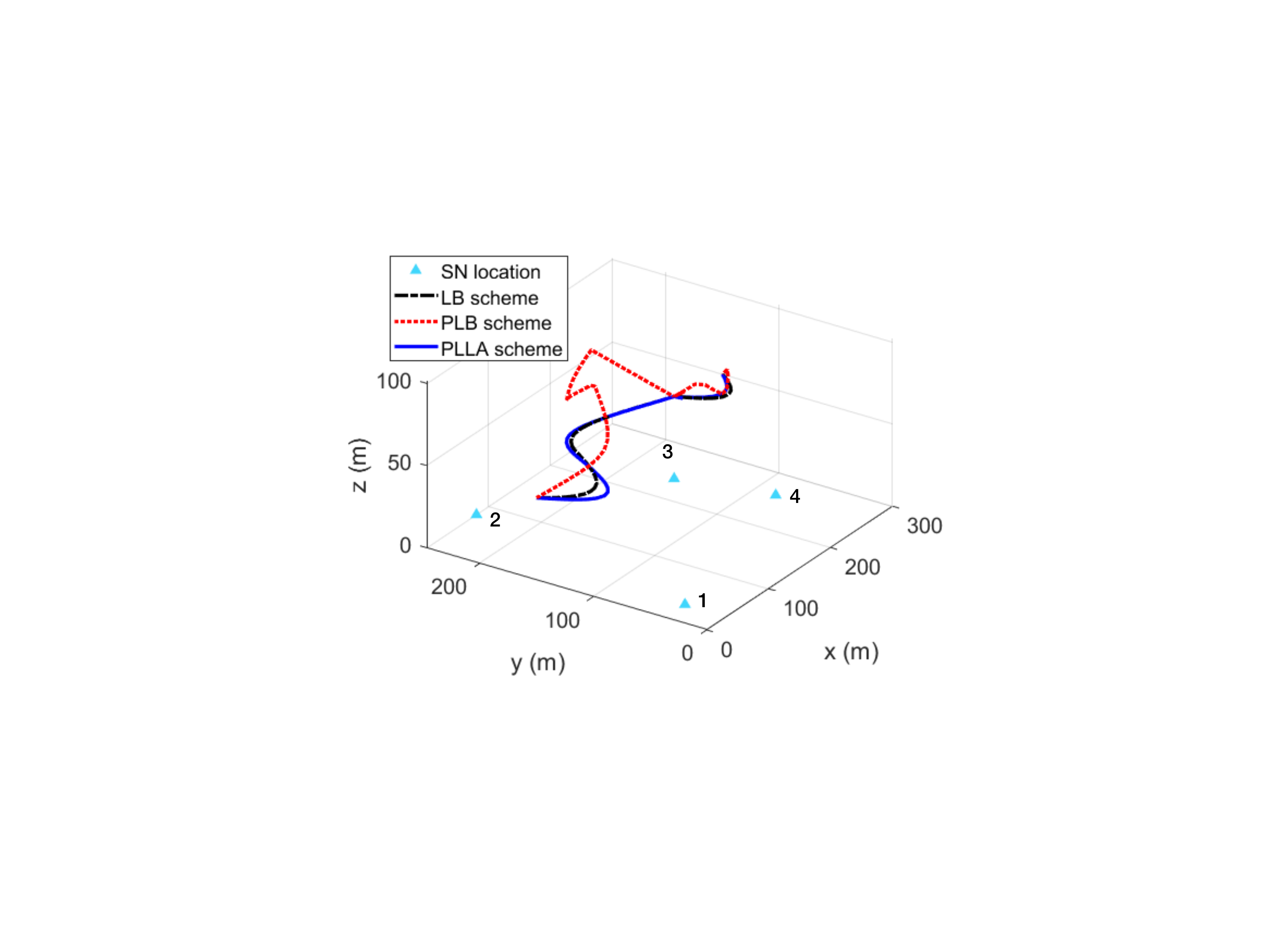}}
\caption{Comparison of the UAV trajectories by  different schemes with $T_0=10.6$ s.}\label{Fig:Traj106}
\end{figure} 

\begin{figure}[t!]
\centering
\subfigure[UAV horizontal trajectory.]{\label{FigHoritraj_time}
\includegraphics[height=5.3cm]{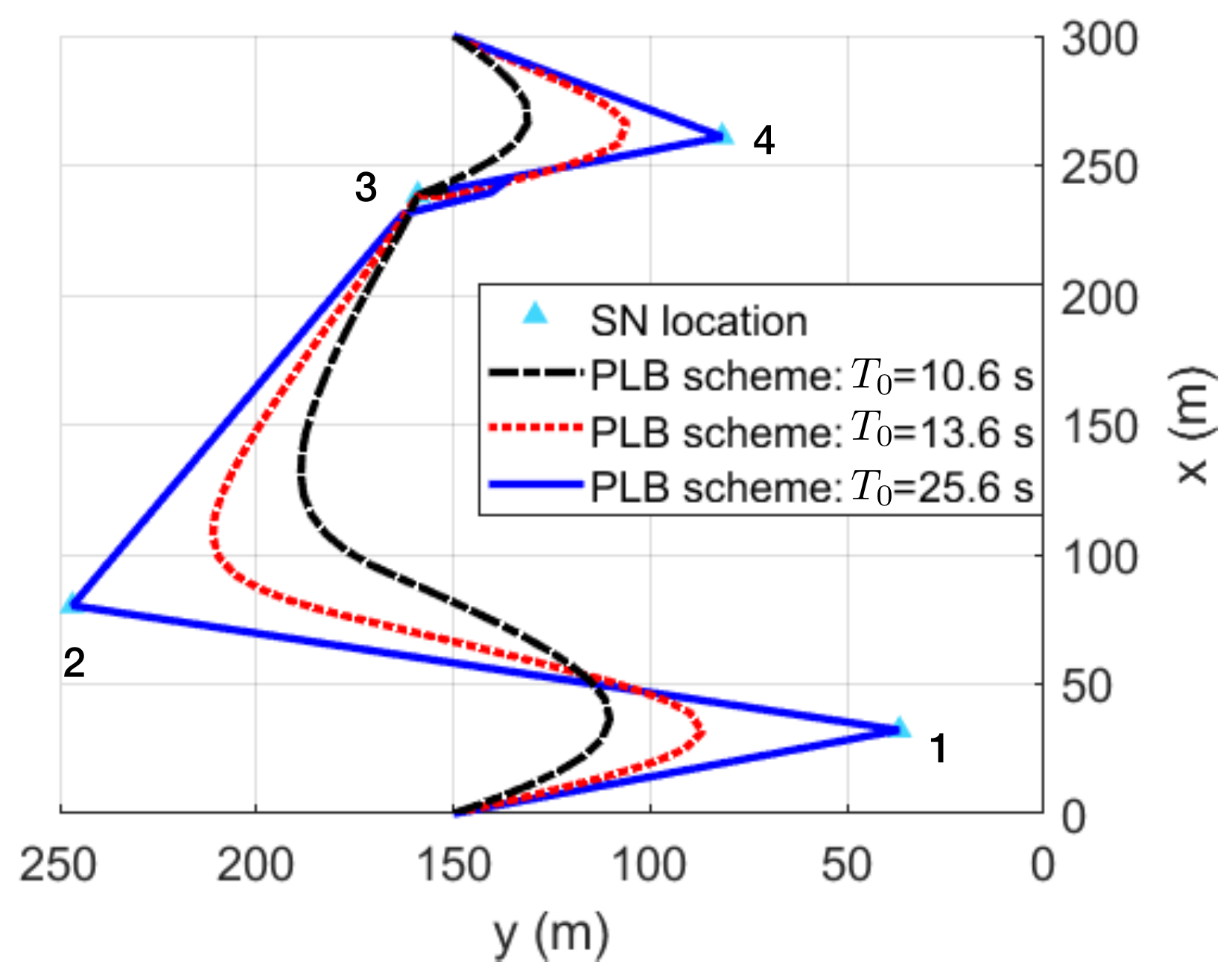}}
\hspace{10mm}
\subfigure[3D UAV trajectory.]{\label{FigVertiTraj_time}
\includegraphics[height=5.3cm]{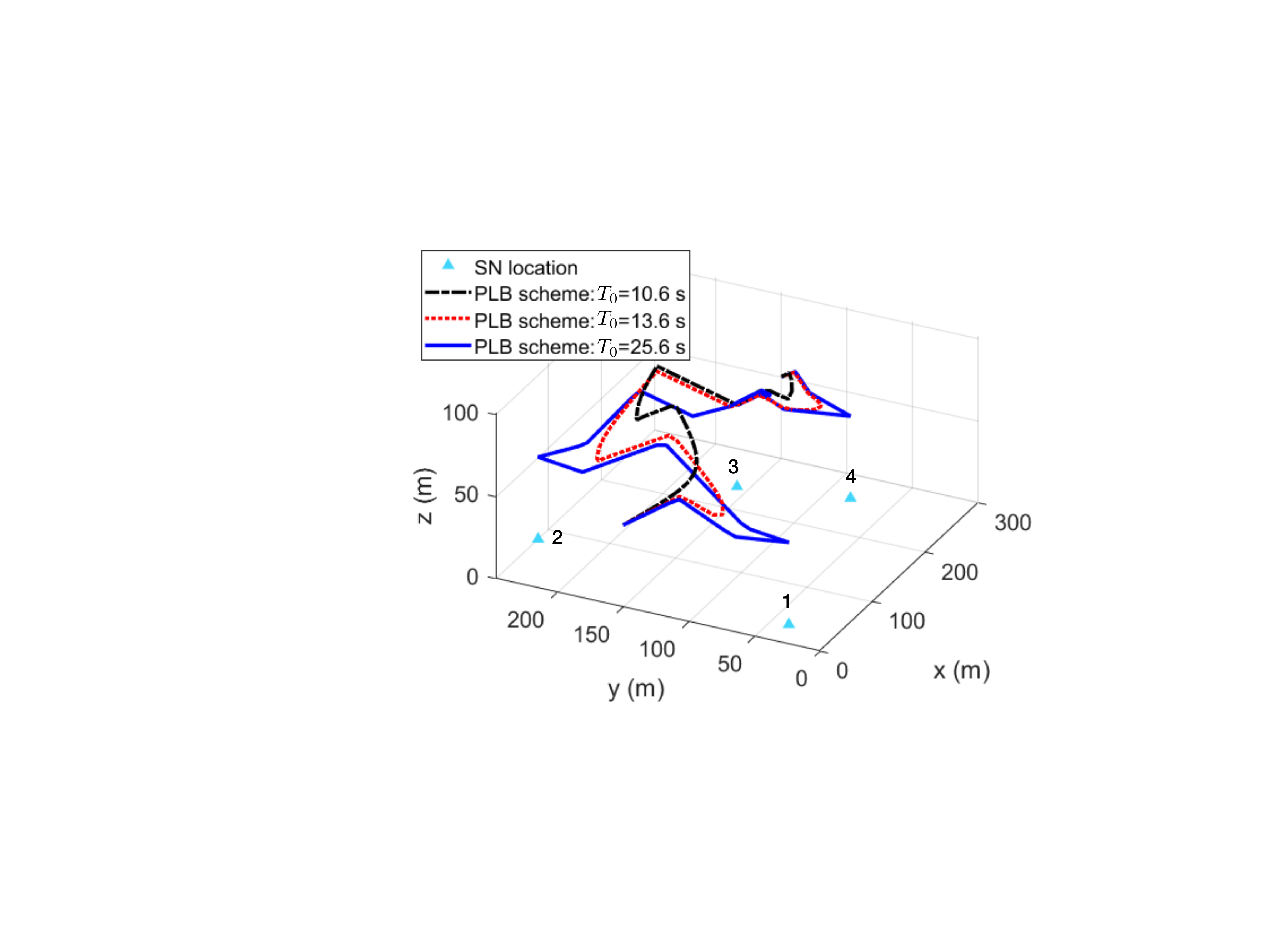}}
\caption{Comparison of the UAV trajectories by the PLB scheme under different UAV flight durations.}\label{Fig:TrajTime}
\end{figure}

In Fig.~\ref{Fig:TrajTime}, we compare the UAV trajectories by the proposed  PLB scheme under different UAV flight durations. It is observed that given a short flight duration (e.g., $T_0= 10.6$ s), the UAV flies at a relatively high altitude to maintain appropriate elevation angles with  SNs for increasing the LoS probabilities. As the flight duration increases, the UAV lowers its altitude and moves closer to the SNs when traveling around them. Specifically, with a sufficiently long flight duration (e.g., $T_0 = 25.6$ s), the UAV can  sequentially hover right above each SN at the minimum altitude for a certain amount of time, which is expected since by this way, the UAV can attain the largest elevation angle with each SN, and at the same time, experience the minimum path loss. 

\begin{figure}[t]
\begin{center}
\includegraphics[height=5.3cm]{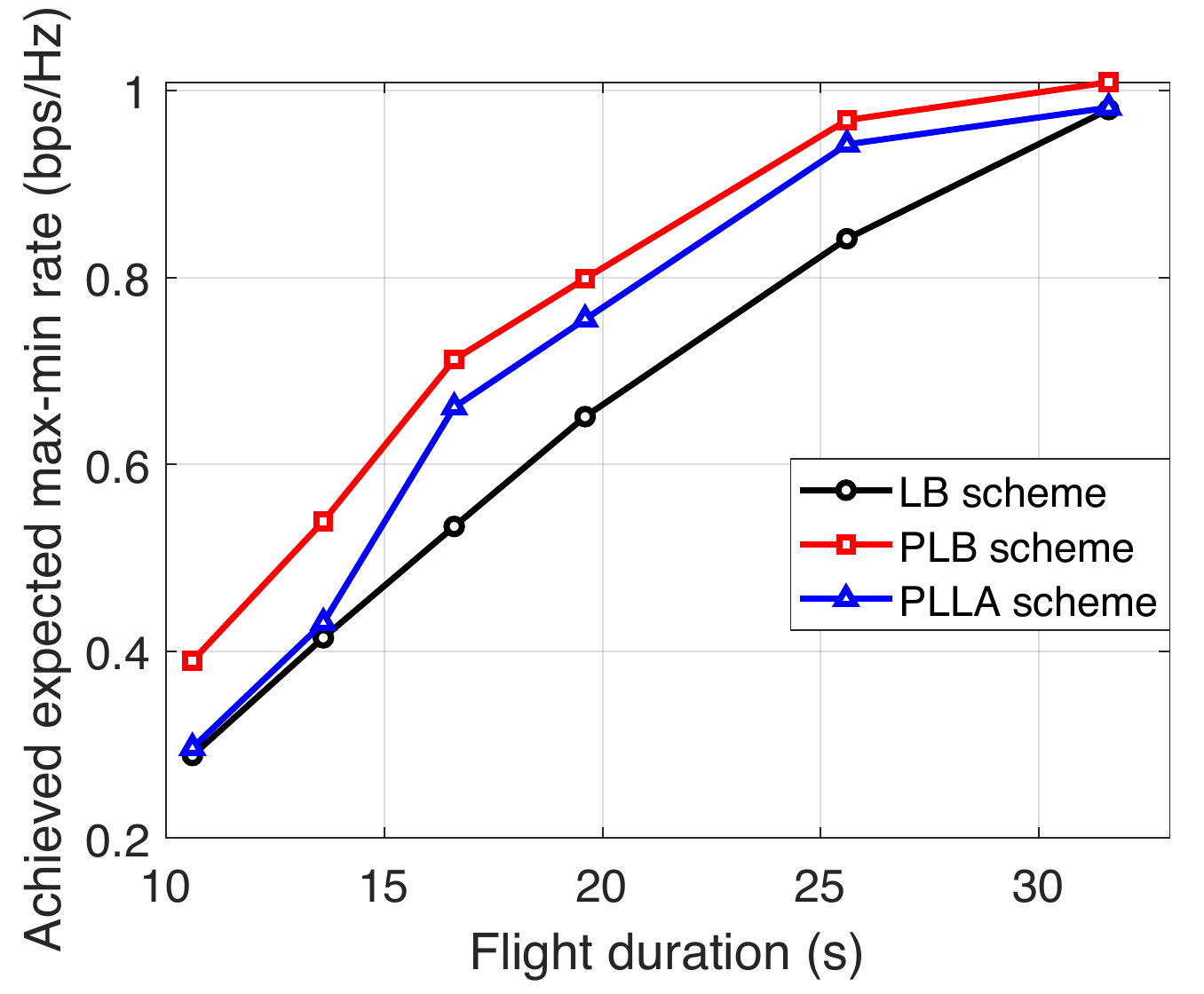}
\caption{Offline phase: achieved expected max-min rate versus flight duration.}
\label{Fig:OffRate}
\end{center}
\end{figure}
Fig.~\ref{Fig:OffRate} plots the curves of the  \emph{achieved expected} max-min rates by different offline schemes averaged over $100$ city realizations versus the UAV flight duration. Note that there is no online adjustment of UAV flying speeds and communication scheduling here. We can observe that the PLLA scheme achieves larger rates than the LB scheme, since it optimizes the UAV horizontal trajectory based on the more accurate   probabilistic LoS channel model instead of the simplified LoS channel model. In addition, the proposed PLB scheme can further improve the rate performance as compared to the PLLA scheme, since it exploits the UAV vertical trajectory to further reduce  the blockage effect. However, the rate performance gain diminishes as the UAV flight duration increases, which is expected since  all the schemes have similar trajectories when the flight duration is sufficiently long, i.e., the UAV will spend more  time on hovering above SNs at its lowest altitude to achieve the highest rates with them, while the time and achieved rates when it is flying  become less significant.

\subsection{Online Phase}\label{Sec:SimuOnline}
Next, we study  the effectiveness of the proposed online design by fixing the UAV path as that  obtained from the offline phase by the PLB scheme. For performance comparison, we consider the following online  schemes: 1) PLB scheme  without  any online adjustment; 2) adaptive communication scheduling (ACS) that only online updates the communication scheduling by solving an LP in each time slot which has a similar form as problem (P11), but without the optimization of UAV flying speeds;  3) our proposed scheme with joint communication scheduling and UAV flying speed adaptation,  thus named as  joint adaptation (JA); 
and 4) optimal joint adaptation (OJA) assuming the ideal case of perfect (non-causal)  UAV-SNs CSI along the offline obtained UAV path, which  jointly optimizes  UAV flying speeds and communication scheduling by solving a similar problem as  problem (P11) with the expected rate replaced by the exact  rate.


\begin{figure}[t!]
\vspace{5pt}
\centering
\subfigure[Online phase: achieved expected max-min rate \newline\indent \quad~ versus flight duration.]{\label{Fig:RateOnline}
\includegraphics[height=5.3cm]{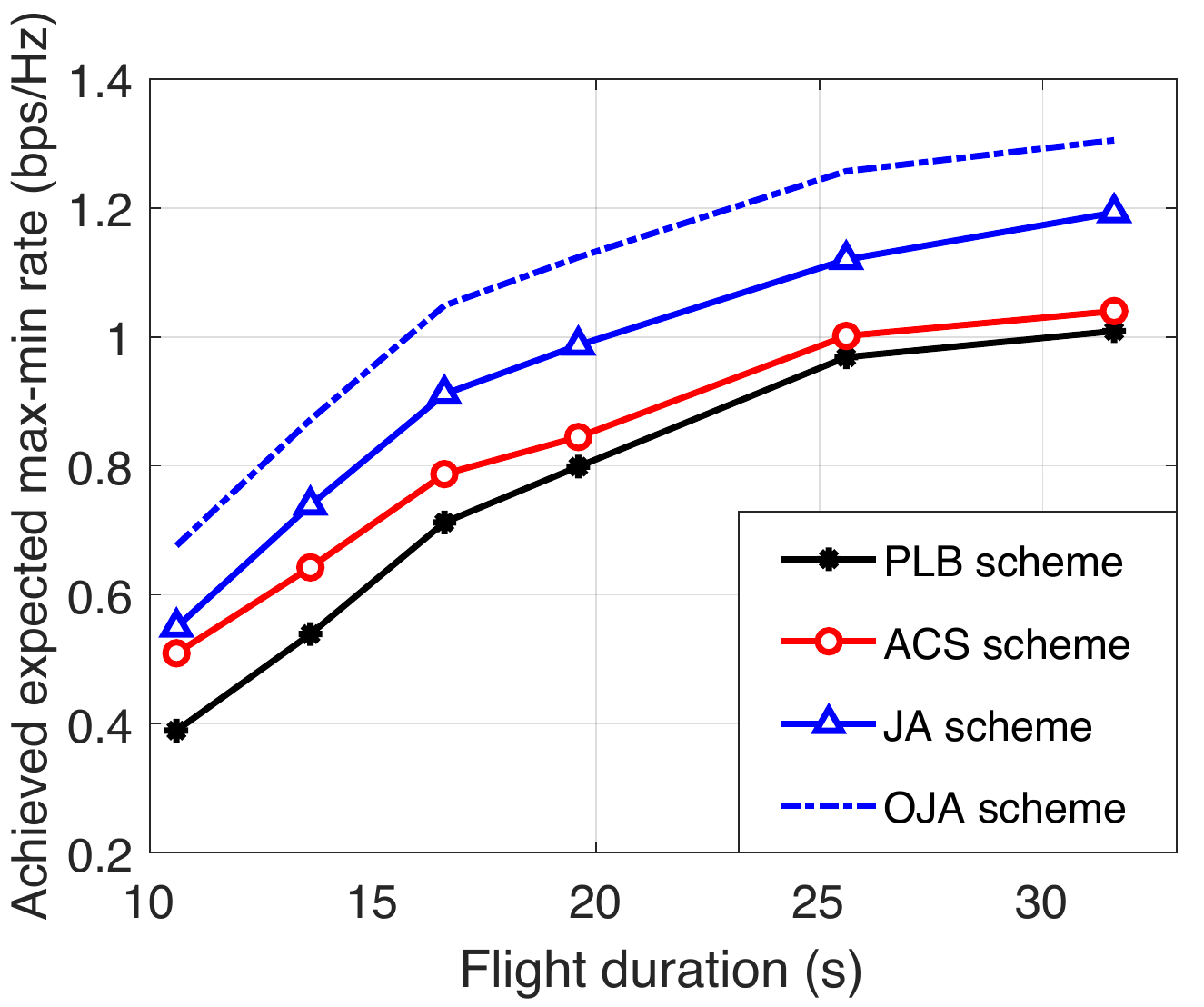}}
\hspace{10mm}
\subfigure[Scheduled transmission rate.]{\label{Fig:SampleTimeRate}
\includegraphics[height=5.3cm]{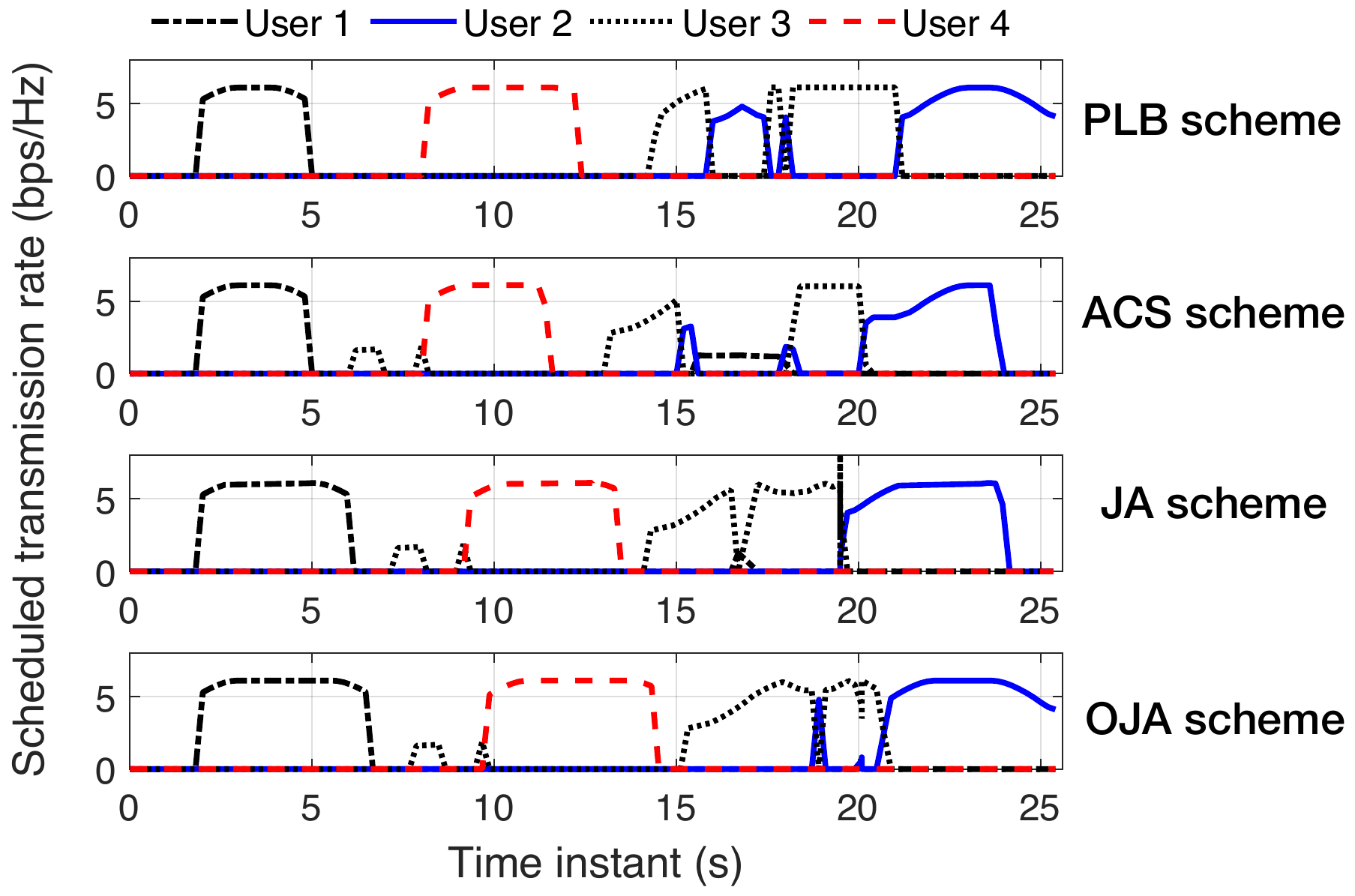}}
\subfigure[Traveling durations of the JA scheme.]{\label{Fig:TravelDuration}
\includegraphics[height=5.3cm]{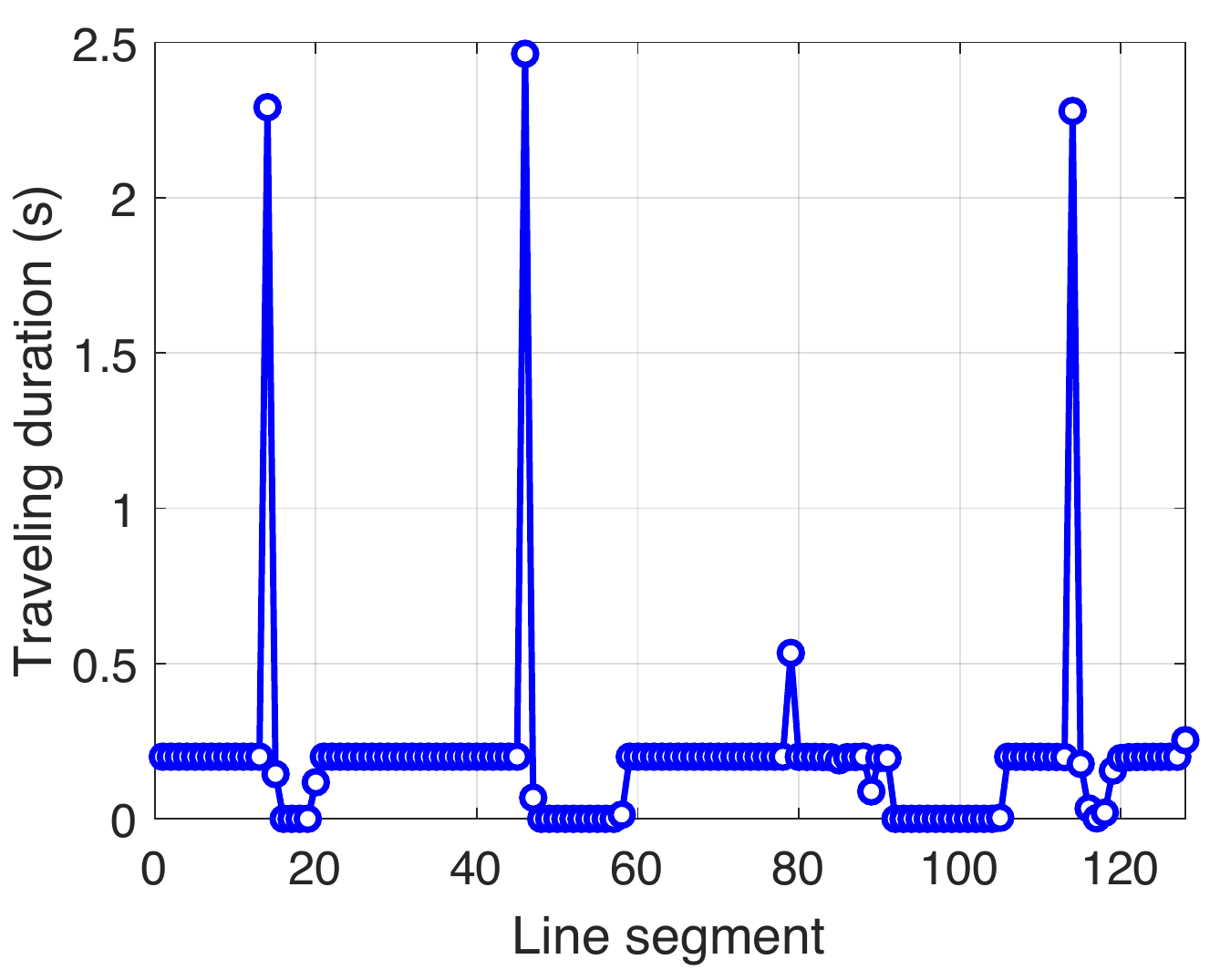}}
\hspace{12mm}
\subfigure[Achieved max-min rate.]{\label{Fig:SamRate}
\includegraphics[height=5.3cm]{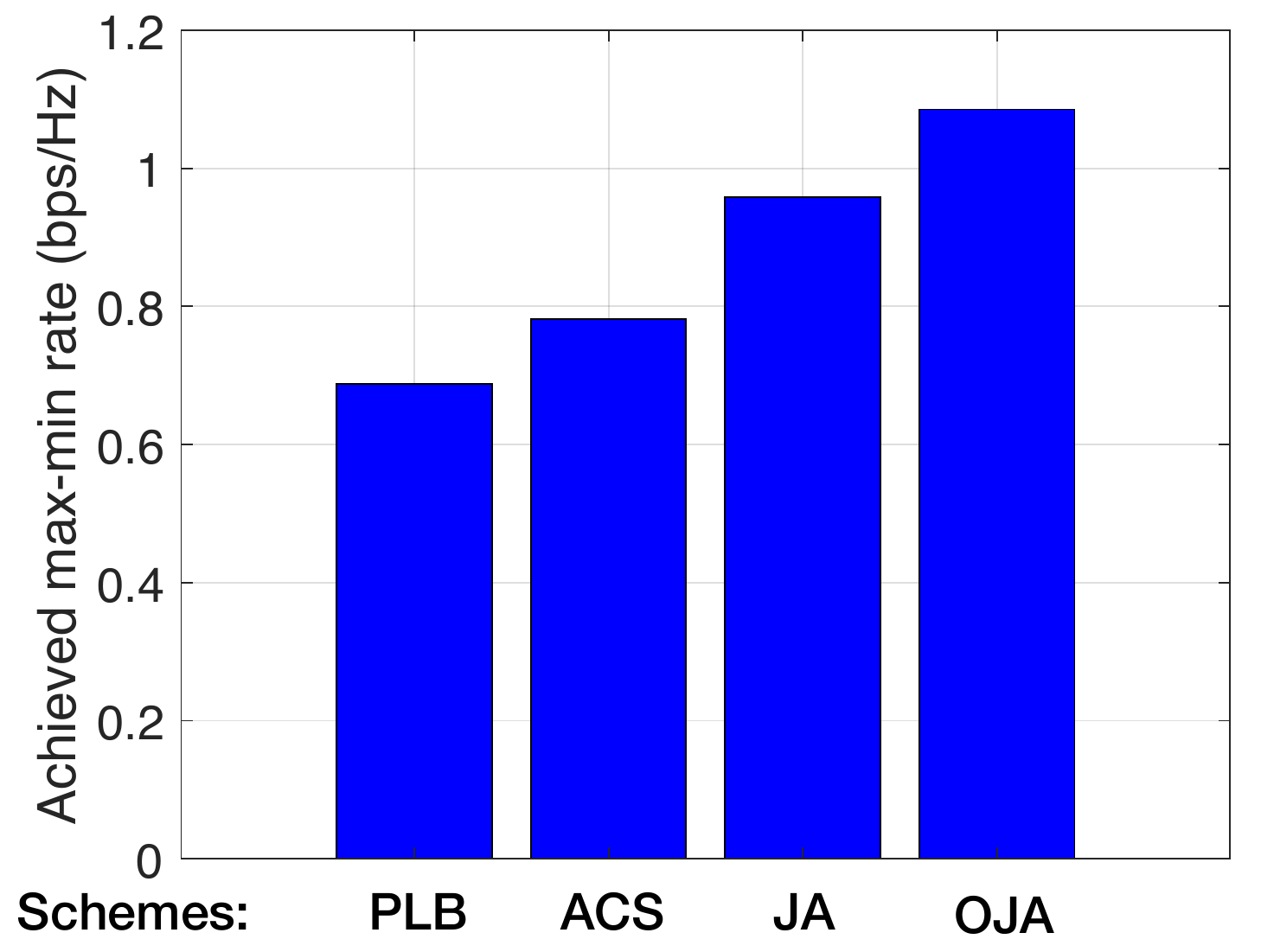}}
\caption{(a): Comparison of the rate performance by different online schemes versus flight duration; (b)-(d): rate performance and traveling durations of different online schemes in a random city realization with $T_0=25.6$ s.}
\label{Fig:OnlineStage}
\end{figure}

In Fig.~\ref{Fig:RateOnline}, we compare the achieved expected max-min rates  by different online schemes over $100$ city realizations versus the flight duration. First, it  is observed that, for all  schemes, the expected max-min rate first increases with the flight duration and then tends to saturate in the regime of long flight duration due to the similar reason given for the offline case (see Fig.~\ref{Fig:OffRate}). Second, although  the simple ACS scheme outperforms the  PLB scheme, its performance gain diminishes with the flight duration. The reason is that, with a sufficiently long flight duration, the dominant rate is contributed by the regime where the UAV flies (nearly) above each SN and thus the UAV has a high likelihood to establish an LoS link with the SN underneath, which renders the online  communication-scheduling adaptation less effective. This limitation, however, can be alleviated by the proposed JA scheme, which  further adjusts  the UAV flying speeds in real time  along the path such that it can spend more time on the line segments that contribute to relatively higher achievable rates, thus sustaining  the rate  gain even for the case of  long flight duration. 
Last, it is worth noting  that  OJA scheme achieves the  highest  max-min rates since it has full and non-causal CSI  along the UAV path, thus providing a performance upper bound for our proposed JA scheme with causal CSI only. 

To further illustrate the proposed JA scheme for the online phase, we compare the rate performance in a random city realization with $T_0=25.6$ s. One can observe from Fig.~\ref{Fig:SampleTimeRate} that the achieved rate from SN $1$ is the performance bottleneck of the PLB scheme. The ACS  scheme can slightly improve the average rate of SN $1$ by scheduling it for transmission  in idle  time slots without significantly compromising the rates of other SNs. In contrast,  our proposed JA scheme can  further improve the max-min rate over the ACS scheme  as shown in Fig~\ref{Fig:SamRate} by online adjusting  the traveling duration at each  line segment based on the instantaneous CSI with SNs (see Fig.~\ref{Fig:TravelDuration}). Again, the OJA scheme can make full use of the non-causal CSI to improve the average rates from all SNs, thus achieving the largest max-min rate.

\begin{figure}[t!]
\vspace{5pt}
\centering
\subfigure[Running time of two online schemes in a random city \newline\indent \quad~realization with $T_0= 19.6$ s. ]{\label{Fig:SamRunTime}
\includegraphics[height=5.3cm]{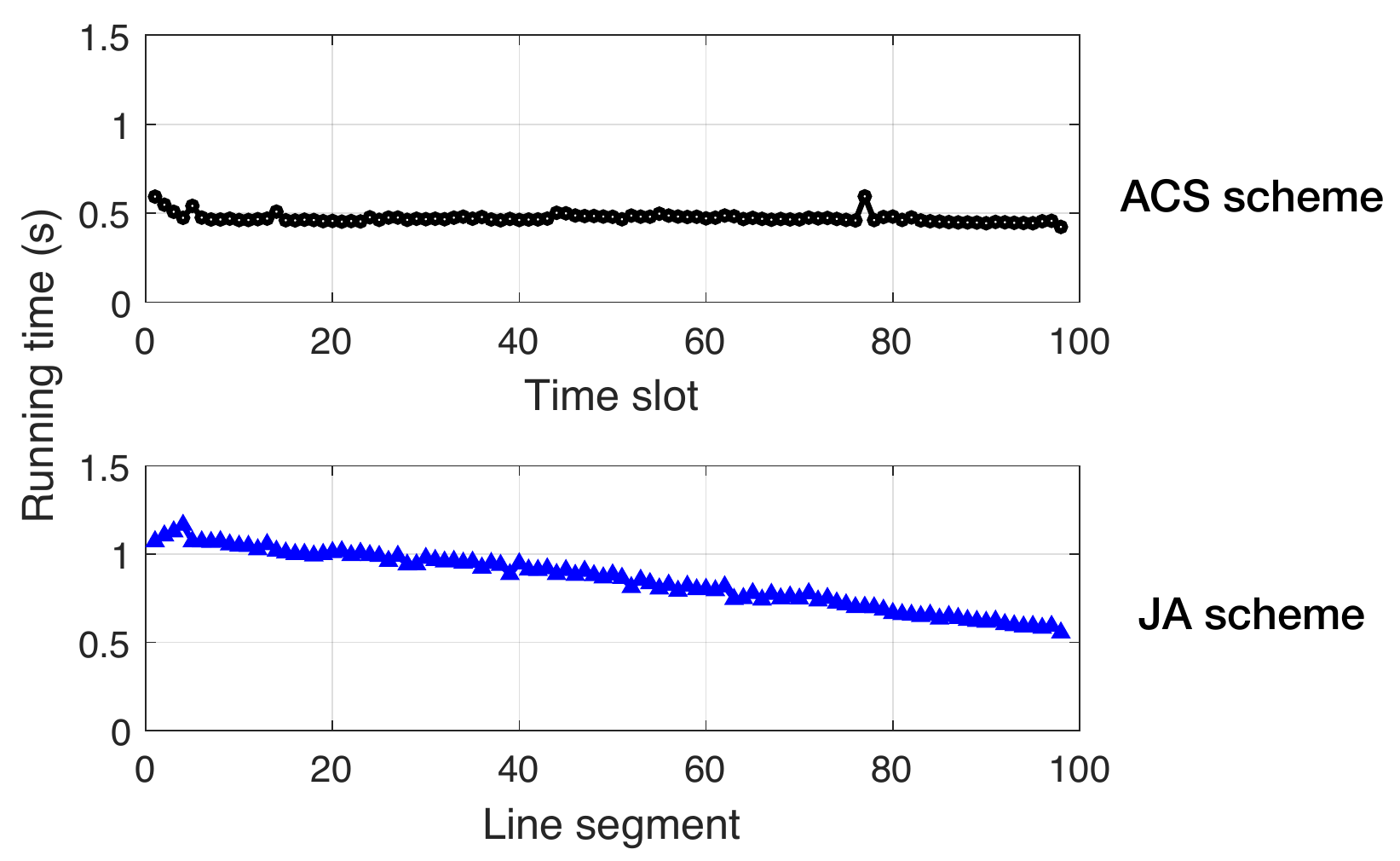}}
\hspace{6mm}
\subfigure[{\color{black}Achieved max-min rate versus number of SNs.}]{\label{FigSN_rate_major}
\includegraphics[height=5.3cm]{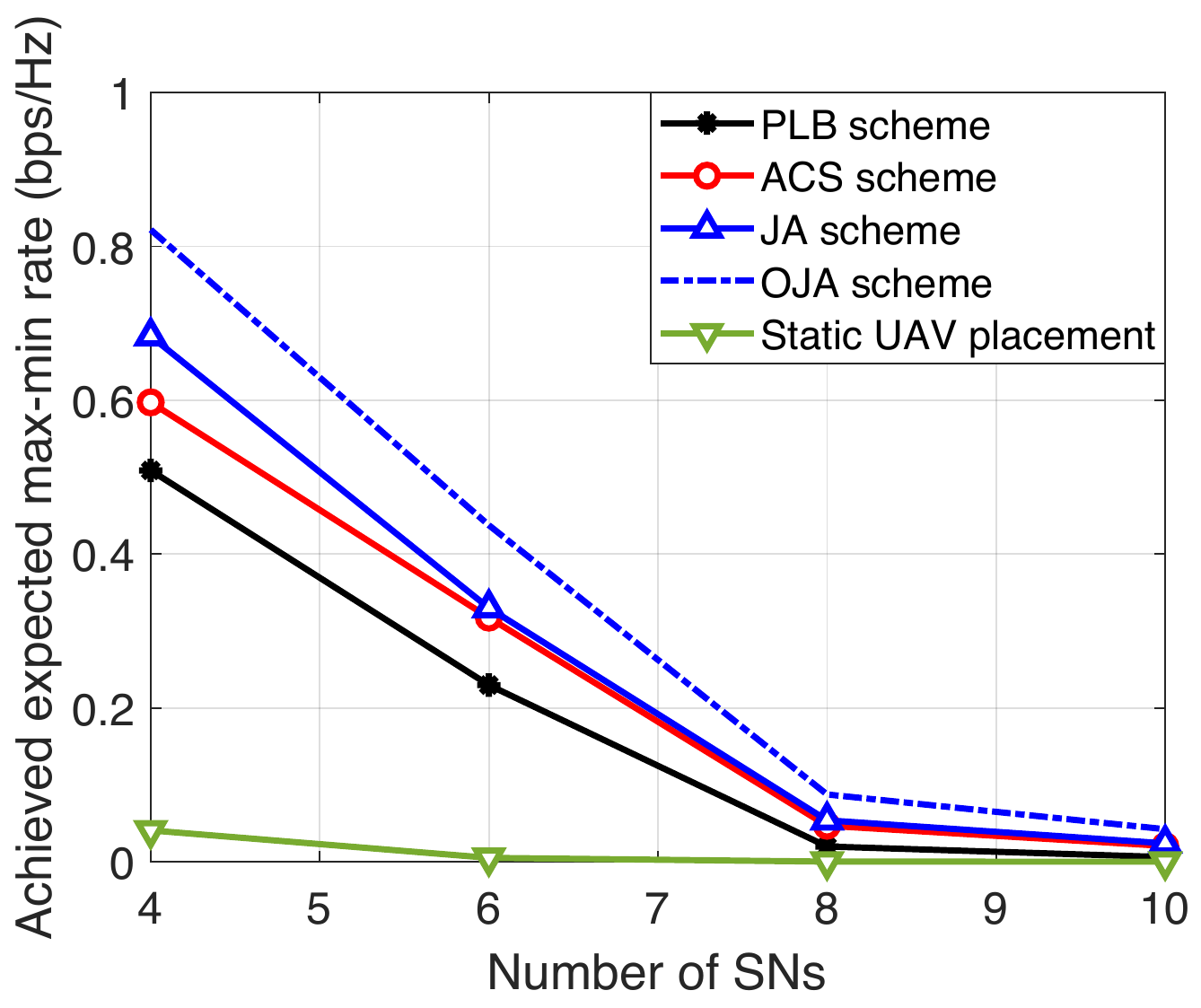}}
\caption{(a): Running time of two online schemes; (b): effects of the number of SNs.}
\label{Fig:OnlineStage}
\end{figure}


To demonstrate the required computational  complexity for implementing  the proposed online designs, we compare the running time of the ACS and JA schemes using Matlab on a computer equipped with Intel Core i5-7500, 3.40 GHz processor, and 8 GB RAM memory.   Fig.~\ref{Fig:SamRunTime} shows the running time of the two schemes for a random city realization with $T_0= 19.6$ s. It is observed that the proposed JA scheme takes slightly longer running time than the ACS scheme due to the joint online optimization of UAV flying speeds and communication scheduling for achieving larger max-min rates (see Fig.~\ref{Fig:RateOnline}). 
Moreover,  the running time of the proposed JA scheme decreases along the UAV path or the increasing line-segment index, which is expected since the required number of optimization variables reduces as the UAV flies to the destination. Overall, the required computation time over each line segment is in the order of a second for both online schemes, which is practically affordable.   

{\color{black}Last, we compare in Fig.~\ref{FigSN_rate_major} the achieved expected max-min rates  by different online schemes versus the number of SNs. To show the gain of trajectory optimization, we consider another benchmark scheme called static UAV placement, for which the UAV's horizontal placement is determined by the K-means clustering algorithm based on SNs' locations, the altitude is optimized by using a similar algorithm as for optimizing the vertical trajectory in this paper, and the communication scheduling is optimized according to the actual UAV-SNs channel states. It is observed that the expected max-min rates of all schemes decrease with the number of SNs, since more SNs share the communication resource and the max-min rate is restricted by the worst-case SN that achieves the lowest rate. Moreover, all the schemes that deploy a flying UAV for data collection significantly outperforms the benchmark with a static UAV. The reason is that the achieved max-min  rate by deploying a static UAV is severely  affected by the UAV-SNs channel states at the optimized UAV's location and can be extremely low if one of the SNs is in the NLoS state, whereas deploying a flying UAV can resolve this issue by exploiting the time-varying channel states with SNs.}

\section{Conclusions}\label{Sec:Conc}

This paper studies the  UAV-enabled WSN by deploying  a UAV to collect  data from distributed SNs in a Manhattan-type city, for which a new probabilistic LoS channel model is constructed and shown in the form of a \emph{generalized logistic} function of the UAV-SN elevation angle. We consider  that the UAV only has the knowledge of SNs' locations and the probabilistic LoS channel model of the environment prior to its flight, while it  can obtain obtain the instantaneous UAV-SNs CSI along its flight.  Our objective is to maximize the minimum (average) data collection rate from all the SNs for the UAV. To this end,  we first formulate a rate maximization problem by jointly optimizing the 3D UAV  trajectory and communication scheduling. Then we propose a novel design method, called \emph{hybrid offline-online optimization}, to obtain a suboptimal solution to it, by leveraging both the statistical and real-time CSI. Our  proposed method decouples the  joint design of UAV trajectory and communication scheduling into two phases: namely, an \emph{offline} phase that determines the UAV path based on the probabilistic LoS channel model, followed by an \emph{online} phase that jointly adjusts the UAV flying speeds and communication scheduling along the offline optimized UAV path based on the instantaneous UAV-SNs CSI and SNs' individual amounts of data  received accumulatively.  Simulation results demonstrate the effectiveness of the proposed hybrid design and reveal key insights into the joint offline-online  3D UAV trajectory and communication scheduling optimization.  

The  proposed hybrid design for UAV-enabled WSNs opens many interesting and important directions that deserve further in-depth investigation, some of which as discussed as follows. First, besides the adopted probabilistic LoS channel model, there exist many other stochastic UAV-ground channel models in the literature, such as the nested segmented ray-tracing channel model and the 3D geometry-based stochastic channel model \cite{zeng2019accessing}. Therefore, it is worth studying whether/how we can apply/extend the proposed hybrid design method to optimize the UAV trajectory and communication scheduling under these channel models. Next, considering the ideal case where the complete UAV-SNs CSI is known \emph{a priori},  the corresponding optimal solution to our studied problem will provide a performance upper bound for our proposed hybrid design with statistical and partial CSI only.  However, how to obtain the optimal solution efficiently is still an open problem and thus needs  further investigation.
 In addition, it is interesting to extend the current work assuming one single data collection operation to the scenario with multiple operations. This setting introduces new design challenges, for example, how to make use of the limited CSI obtained from the previous operations to design the UAV trajectory in the subsequent operations for improving their rate performance. Last, in practice, to further enhance the system rate performance, multiple UAVs can be deployed to  collect data from SNs. In this scenario, it is critical to design the cooperation schemes among the UAVs in both the offline and online phases by leveraging the shared CSI for achieving different objectives such as the maximum min-rate and the maximum energy efficiency.
%
%
%

\appendix
\subsection{Simulation-Based Modeling  for LoS Probability}\label{App:SimuLoS}
Consider a Manhattan-type city where the whole area is partitioned into uniform square grids  in the average  size of  buildings therein which can be obtained based on  typical city parameters \cite{data2003prediction}.  We first randomly and uniformly generate the buildings in the square grids with random heights following the Rayleigh distribution (see Fig.~\ref{FigCity}). Next, a set of SNs are randomly generated on the area unoccupied by the buildings. By using a similar simulation method as in \cite{holis2008elevation}, we then obtain the LoS/NLoS channel states under different UAV-SN elevation angles. Specifically, to account for the effect of UAV horizontal location on the LoS probability, we fix the UAV altitude and change the UAV  horizontal location according to the elevation angle (from $5^{\circ}$ to $90^{\circ}$ with a step size of $5^{\circ}$), whereas  the azimuth angle ranges from $0^{\circ}$ to $360^{\circ}$ with a step size of $30^{\circ}$ for each elevation angle. We repeat this procedure for different UAV altitudes ranging from $30$ m to $300$ m with a step size of $30$ m for characterizing the effect of UAV altitude on the LoS probability. Last, the LoS probability is calculated based on the obtained LoS/NLoS states over $200$ randomly-generated cities for each type of environment (see Fig.~\ref{Fig:LoSProb}).

\vspace{-8pt}
\subsection{Proof of Lemma~\ref{Lem:ConvexRateFunc}}\label{App:ConvexRateFunc}
Define $\xi(x,y)\overset{\triangle}{=}\(B_3+\frac{B_4}{x}\r)\ln\l(1+\frac{\gamma}{y^{\alpha/2}}\r)$. We first prove the convexity of $\xi(x,y)$ by the definition of convex functions. To this end, we first derive that first-order derivatives of $\xi(x,y)$ w.r.t. $x$ and $y$, as follows.
\begin{align}
\xi_x(x,y)=-\frac{B_4}{x^2}\ln\l(1+\frac{\gamma}{y^{\alpha/2}}\r),\quad
\xi_y(x,y)=-\l(B_3+\frac{B_4}{x}\r)\frac{\gamma \alpha/2}{y(y^{\alpha/2}+\gamma)}.
\end{align} 
Then, the Hessian of $\xi(x,y)$ is
\begin{equation}
\bigtriangledown^2 \xi(x,y)=\l[
\begin{array}{cc}
\dfrac{2B_4}{x^3}\ln\l(1+\dfrac{\gamma}{y^{\alpha/2}}\r)& \dfrac{B_4 \gamma \alpha/2}{x^2y(y^{\alpha/2}+\gamma)}\\
 \dfrac{B_4 \gamma \alpha/2}{x^2y(y^{\alpha/2}+\gamma)} & \l(B_3+\dfrac{B_4}{x}\r) \dfrac{\gamma\alpha/2[(\alpha/2+1)y^{\alpha/2}+\gamma] }{y^2(y^{\alpha/2}+\gamma)^2}
 \end{array}\r].
\end{equation}
Given $\alpha\ge2$, for any $\mathbf{t}=[t_1, t_2]^T$, we have 
\begin{align*}
&\mathbf{t}^T \bigtriangledown^2 \xi(x,y) \mathbf{t}\nn\\
=&t_1^2\l(\frac{2B_4}{x^3}\ln\l(1+\frac{\gamma}{y^{\alpha/2}}\r)\r)+t_2^2\l(\l(B_3+\frac{B_4}{x}\r) \frac{\gamma \alpha/2[(\alpha/2+1)y^{\alpha/2}+\gamma]}{y^2(y^{\alpha/2}+\gamma)^2}\r)+2t_1t_2\l(\frac{B_4 \gamma \alpha/2}{x^2y(y^{\alpha/2}+\gamma)}\r)\nn \\
\overset{(i)}{\ge}&t_1^2\l(\frac{2B_4}{x^3}\frac{\gamma}{y^{\alpha/2}+\gamma}\r)+t_2^2\l(\l(B_3+\frac{B_4}{x}\r) \frac{\gamma(2y^{\alpha/2}+\gamma) }{y^2(y^{\alpha/2}+\gamma)^2}\r)+2t_1t_2\l(\frac{B_4 \gamma}{x^2y(y^{\alpha/2}+\gamma)}\r)\nn \\
=&\frac{2t_1^2 B_4 \gamma y^2(y^{\alpha/2}+\gamma )+t_2^2\gamma  x^3 (B_3+\frac{B_4}{x}) (2y^{\alpha/2}+\gamma )+2t_1t_2B_4\gamma xy(y^{\alpha/2}+\gamma)}{x^3 y^2(y^{\alpha/2}+\gamma )^2}\nn\\
=&\frac{\gamma[(y^{\alpha/2}+\gamma)B_4(t_1y+t_2x)^2+(y^{\alpha/2}+\gamma)(B_4 t_1^2y+t_2^2x^3B_3)+t_2^2y^{\alpha/2}x^3(B_3+\frac{B_4}{x})] }{x^3 y^2(y^{\alpha/2}+\gamma )^2}\nn\\
\ge&\frac{\gamma[(y^{\alpha/2}+\gamma)B_4(t_1y+t_2x)^2+(y^{\alpha/2}+\gamma)\min(t_1^2y^2x,t_2^2x^3)(B_3+\frac{B_4}{x})+t_2^2y^{\alpha/2}x^3(B_3+\frac{B_4}{x})] }{x^3 y^2(y^{\alpha/2}+\gamma )^2}>0,
\label{Eq:Hassian}
\end{align*} 
for $x> 0$ and $y> 0$, where $(i)$ is due to $\alpha\ge2$ and $\ln(1+\frac{1}{a})\ge\frac{1}{a+1}$ for $a>0$. Therefore, $\xi(x,y)$ is a convex function, leading to the convexity of $\psi(x,y)$.
\subsection{Proof of Lemma~\ref{Lem:RateBound}}\label{App:RateBound}
\vspace{-5pt}
Using Lemma~\ref{Lem:ConvexRateFunc}, it can be proved that $\tilde{\psi}(x,y)=(B_3+\frac{B_4}{X+x})\log_2\l(1+\frac{\gamma}{Y+y}\r)$ is a convex function w.r.t. $x> -X$ and $y> -Y$. Then using the SCA technique, for any given $x_0$ and $y_0$, we have $\tilde{\psi}(x,y)\ge  \tilde{\psi}(x_0,y_0)+\tilde{\psi}_x(x_0,y_0)(x-x_0)+\tilde{\psi}_y(x_0,y_0)(y-y_0), \forall x, y$, where
\vspace{-3pt}
\begin{align*}
\tilde{\psi}_x(x_0,y_0)&=\frac{-B_4(\log_2{e})}{(X+x_0)^2}\ln\l(1+\frac{\gamma}{(Y+y_0)^{\alpha/2}}\r), \\
\tilde{\psi}_y(x_0,y_0)&=-\l(B_3+\frac{B_4}{X+x_0}\r)\frac{\gamma \alpha/2 (\log_2{e})}{(Y+y_0)((Y+y_0)^{\alpha/2}+\gamma)}.
\end{align*}
By letting $x_0=0$ and $y_0=0$, we can obtain
\vspace{-4pt}
\begin{align*}
&\l(B_3+\frac{B_4}{X+x}\r)\log_2\l(1+\frac{\gamma}{Y+y}\r)\nn\\
&\ge\l(B_3+\frac{B_4}{X}\r)\log_2\l(1+\frac{\gamma}{Y}\r)-
\frac{B_4(\log_2{e})}{X^2}\ln\l(1+\frac{\gamma}{Y^{\alpha/2}}\r) x-
\l(B_3+\frac{B_4}{X}\r)\frac{\gamma \alpha/2 (\log_2{e})}{Y(Y^{\alpha/2}+\gamma)}y.
\end{align*}
Last, by letting $\gamma=\gamma_k$,    $X=1+ e^{-\hat{\phi}_{k,n}}$, $x= e^{-{\phi}_{k,n}}- e^{-\hat{\phi}_{k,n}}$, $\hat{\phi}_{k,n}=B_1+B_2\hat{\theta}_{k,n}$, $Y=|| \hat{\mathbf{q}}_n-\mathbf{w}_k||^2+z_n^2$, and $y= || \mathbf{q}_n-\mathbf{w}_k||^2-|| \hat{\mathbf{q}}_n-\mathbf{w}_k||^2$, we thus derive Lemma~\ref{Lem:RateBound} where $\hat{\bar{r}}^{\rm{L}}_{k,n}=(B_3+\frac{B_4}{X+x})\log_2\l(1+\frac{\gamma}{Y+y}\r)$, $\hat{\Omega}_{k,n}=\frac{B_4(\log_2{e})}{X^2}\ln\l(1+\frac{\gamma}{Y^{\alpha/2}}\r)$, and $\hat{\Psi}_{k,n}=\l(B_3+\frac{B_4}{X}\r)\frac{\gamma \alpha/2 (\log_2{e})}{Y(Y^{\alpha/2}+\gamma)}$.
\vspace{-10pt}

\begin{thebibliography}{10}
\providecommand{\url}[1]{#1}
\csname url@samestyle\endcsname
\providecommand{\newblock}{\relax}
\providecommand{\bibinfo}[2]{#2}
\providecommand{\BIBentrySTDinterwordspacing}{\spaceskip=0pt\relax}
\providecommand{\BIBentryALTinterwordstretchfactor}{4}
\providecommand{\BIBentryALTinterwordspacing}{\spaceskip=\fontdimen2\font plus
\BIBentryALTinterwordstretchfactor\fontdimen3\font minus
  \fontdimen4\font\relax}
\providecommand{\BIBforeignlanguage}[2]{{%
\expandafter\ifx\csname l@#1\endcsname\relax
\typeout{** WARNING: IEEEtran.bst: No hyphenation pattern has been}%
\typeout{** loaded for the language `#1'. Using the pattern for}%
\typeout{** the default language instead.}%
\else
\language=\csname l@#1\endcsname
\fi
#2}}
\providecommand{\BIBdecl}{\relax}
\BIBdecl

\bibitem{you2019globecom}
C.~You, X.~Peng, and R.~Zhang, ``{3D} trajectory design for {UAV}-enabled data
  harvesting in probabilistic {LoS} channel,'' in \emph{Proc. IEEE Global
  Commun. Conf. (Globecom)}, Dec. 2019.

\bibitem{zeng2019accessing}
Y.~Zeng, Q.~Wu, and R.~Zhang, ``Accessing from the sky: {A} tutorial on {UAV}
  communications for {5G} and beyond and beyond,'' \emph{Proc. IEEE}, vol. 107,
  no.~12, pp. 2327--2375, Dec. 2019.

\bibitem{zeng2016wireless}
Y.~Zeng, R.~Zhang, and T.~J. Lim, ``Wireless communications with unmanned
  aerial vehicles: {Opportunities} and challenges,'' \emph{IEEE Commun. Mag.},
  vol.~54, no.~5, pp. 36--42, May 2016.

\bibitem{wu2018joint}
Q.~Wu, Y.~Zeng, and R.~Zhang, ``Joint trajectory and communication design for
  multi-{UAV} enabled wireless networks,'' \emph{IEEE Trans. Wireless Commun.},
  vol.~17, no.~3, pp. 2109--2121, Mar. 2018.

\bibitem{lyu2017placement}
J.~Lyu, Y.~Zeng, R.~Zhang, and T.~J. Lim, ``Placement optimization of
  {UAV}-mounted mobile base stations,'' \emph{IEEE Commun. Lett.}, vol.~21,
  no.~3, pp. 604--607, Mar. 2017.

\bibitem{mozaffari2016efficient}
M.~Mozaffari, W.~Saad, M.~Bennis, and M.~Debbah, ``Efficient deployment of
  multiple unmanned aerial vehicles for optimal wireless coverage.'' \emph{IEEE
  Commun. Lett.}, vol.~20, no.~8, pp. 1647--1650, Aug. 2016.

\bibitem{mozaffari2016unmanned}
------, ``Unmanned aerial vehicle with underlaid device-to-device
  communications: Performance and tradeoffs and tradeoffs,'' \emph{IEEE Trans.
  Wireless Commun.}, vol.~15, no.~6, pp. 3949--3963, Jun. 2016.

\bibitem{bor2016efficient}
R.~I. Bor-Yaliniz, A.~El-Keyi, and H.~Yanikomeroglu, ``Efficient 3-{D}
  placement of an aerial base station in next generation cellular networks,''
  in \emph{Proc. IEEE Intl. Conf. Commun. (ICC)}, May 2016.

\bibitem{zhang2018cellular}
S.~Zhang, Y.~Zeng, and R.~Zhang, ``Cellular-enabled {UAV} communication: {A}
  connectivity-constrained trajectory optimization perspective,'' \emph{IEEE
  Trans. Commun.}, vol.~67, no.~3, pp. 2580--2604, Mar. 2019.

\bibitem{zeng2019cellular}
Y.~Zeng, J.~Lyu, and R.~Zhang, ``Cellular-connected {UAV}: {Potential},
  challenges, and promising technologies,'' \emph{IEEE Wireless Commu.},
  vol.~26, no.~1, pp. 120--127, Feb. 2019.

\bibitem{lyu2019network}
J.~Lyu and R.~Zhang, ``Network-connected {UAV}: {3D} system modeling and
  coverage performance analysis,'' \emph{IEEE IoT J.}, vol.~6, no.~4, pp.
  7048--7060, Aug. 2019.

\bibitem{zeng2016throughput}
Y.~Zeng, R.~Zhang, and T.~J. Lim, ``Throughput maximization for {UAV}-enabled
  mobile relaying systems,'' \emph{IEEE Trans. Commun.}, vol.~64, no.~12, pp.
  4983--4996, Dec. 2016.

\bibitem{chen2018local}
J.~Chen and D.~Gesbert, ``Efficient local map search algorithms for the
  placement of flying relays,'' \emph{to appear in IEEE Trans. Wireless
  Commun.}, {[Online]. Available: https://arxiv.org/pdf/1801.03595.pdf}.

\bibitem{you20193d}
C.~You and R.~Zhang, ``{3D} trajectory optimization in {Rician} fading for
  {UAV}-enabled data harvesting,'' \emph{IEEE Trans. Wireless Commun.},
  vol.~18, no.~6, pp. 3192--3207, Jun. 2019.

\bibitem{zhan2018energy}
C.~Zhan, Y.~Zeng, and R.~Zhang, ``Energy-efficient data collection in {UAV}
  enabled wireless sensor network,'' \emph{IEEE Wireless Commmu. Lett.},
  vol.~7, no.~3, pp. 328--331, Jun. 2018.

\bibitem{ebrahimi2018uav}
D.~Ebrahimi, S.~Sharafeddine, P.-H. Ho, and C.~Assi, ``{UAV}-aided
  projection-based compressive data gathering in wireless sensor networks,''
  \emph{IEEE IoT J.}, vol.~6, no.~2, pp. 1893--1905, Apr. 2019.

\bibitem{gong2018flight}
J.~Gong, T.-H. Chang, C.~Shen, and X.~Chen, ``Flight time minimization of {UAV}
  for data collection over wireless sensor networks,'' \emph{IEEE J. Sel. Areas
  Commun.}, vol.~36, no.~9, pp. 1942--1954, Sep. 2018.
{\color{black}
\bibitem{liu2018age}
J.~Liu, X.~Wang, B.~Bai, and H.~Dai, ``Age-optimal trajectory planning for
  {UAV}-assisted data collection,'' in \emph{Proc. IEEE Int. Conf. Comput.
  Commun. (INFOCOM) Workshops}, 2018.

\bibitem{abd2018average}
M.~A. Abd-Elmagid and H.~S. Dhillon, ``Average peak age-of-information
  minimization in {UAV}-assisted {IoT} networks,'' \emph{IEEE Trans. Veh.
  Techn.}, vol.~68, no.~2, pp. 2003--2008, Feb. 2018.

\bibitem{abd2019deep}
M.~A. Abd-Elmagid, A.~Ferdowsi, H.~S. Dhillon, and W.~Saad, ``Deep
  reinforcement learning for minimizing age-of-information in {UAV}-assisted
  networks,'' {[Online]. Available: https://arxiv.org/pdf/1905.02993.pdf}.}

\bibitem{azari2018ultra}
M.~M. Azari, F.~Rosas, K.-C. Chen, and S.~Pollin, ``Ultra reliable {UAV}
  communication using altitude and cooperation diversity,'' \emph{IEEE Trans.
  Commun.}, vol.~66, no.~1, pp. 330--344, Jan. 2018.

\bibitem{al2014optimal}
A.~Al-Hourani, S.~Kandeepan, and S.~Lardner, ``Optimal {LAP} altitude for
  maximum coverage,'' \emph{IEEE Wireless Commu. Lett.}, vol.~3, no.~6, pp.
  569--572, Dec. 2014.

\bibitem{esrafilian2018learning}
O.~Esrafilian, R.~Gangula, and D.~Gesbert, ``Learning to communicate in
  {UAV}-aided wireless networks: {Map}-based approaches,'' \emph{IEEE IoT J.},
  vol.~6, no.~2, pp. 1791--1802, Apr. 2019.

\bibitem{zeng2019energy}
Y.~Zeng, J.~Xu, and R.~Zhang, ``Energy minimization for wireless communication
  with rotary-wing {UAV},'' \emph{IEEE Trans. Wireless Commun.}, vol.~18,
  no.~4, pp. 2329--2345, Apr. 2019.

\bibitem{zeng2019path}
Y.~Zeng and X.~Xu, ``Path design for cellular-connected {UAV} with
  reinforcement learning,'' in \emph{Proc. IEEE Global Commun. Conf.
  (Globecom)}, Dec. 2019.

\bibitem{liu2018energy}
C.~H. Liu, Z.~Chen, J.~Tang, J.~Xu, and C.~Piao, ``Energy-efficient {UAV}
  control for effective and fair communication coverage: {A} deep reinforcement
  learning approach,'' \emph{IEEE J. Sel. Areas Commun.}, vol.~36, no.~9, pp.
  2059--2070, Sep. 2018.

\bibitem{challita2019interference}
U.~Challita, W.~Saad, and C.~Bettstetter, ``Interference management for
  cellular-connected {UAVs}: {A} deep reinforcement learning approach,''
  \emph{IEEE Trans. Wireless Commun.}, vol.~19, no.~4, pp. 2125--2140, Apr.
  2019.

\bibitem{liu2018trajectory}
X.~Liu, Y.~Liu, Y.~Chen, and L.~Hanzo, ``Trajectory design and power control
  for multi-{UAV} assisted wireless networks: A machine learning approach,''
  \emph{IEEE Trans. Veh. Techn.}, vol.~68, no.~8, pp. 7957--7969, Aug. 2019.

\bibitem{kuwata2003real}
Y.~Kuwata, ``Real-time trajectory design for unmanned aerial vehicles using
  receding horizontal control,'' Ph.D. dissertation, Massachusetts Institute of
  Technology, 2003.

\bibitem{barraquand1997random}
J.~Barraquand, L.~Kavraki, J.-C. Latombe, R.~Motwani, T.-Y. Li, and
  P.~Raghavan, ``A random sampling scheme for path planning,'' \emph{Int. J.
  Robotics Research}, vol.~16, no.~6, pp. 759--774, 1997.

\bibitem{data2003prediction}
P.~Data, ``Prediction methods required for the design of terrestrial broadband
  millimetric radio access systems operating in a frequency range of about
  20-50 {GHz},'' \emph{Draft New Reco. ITU-R P.[DOC. 3/47], Working Party K},
  vol.~3, 2003.

\bibitem{CVX}
\BIBentryALTinterwordspacing
M.~Grant and S.~Boyd, \emph{{CVX}: {Matlab} software for disciplined convex
  programming, version 2.1}. [Online]. Available: \url{http://cvxr.com/cvx}
\BIBentrySTDinterwordspacing

\bibitem{ben2001lectures}
A.~Ben-Tal and A.~Nemirovski, \emph{Lectures on modern convex optimization:
  analysis, algorithms, and engineering applications}.\hskip 1em plus 0.5em
  minus 0.4em\relax SIAM, 2001, vol.~2.

\bibitem{bertsekas1995dynamic}
D.~P. Bertsekas, \emph{Dynamic programming and optimal control}.\hskip 1em plus
  0.5em minus 0.4em\relax Athena scientific Belmont, MA, 1995, vol.~1, no.~2.

\bibitem{holis2008elevation}
J.~Holis and P.~Pechac, ``Elevation dependent shadowing model for mobile
  communications via high altitude platforms in built-up areas,'' \emph{IEEE
  Trans. Ant. Propa.}, vol.~56, no.~4, pp. 1078--1084, Apr. 2008.

\end{thebibliography}

\end{document}